\documentclass[11pt]{iopart}

\usepackage[T1]{fontenc}
\usepackage[utf8]{inputenc}
\usepackage[english]{babel}
\usepackage[noadjust]{cite}
\usepackage{booktabs}

\expandafter\let\csname equation*\endcsname\relax
\expandafter\let\csname endequation*\endcsname\relax
\usepackage{amsmath,amsfonts,amssymb,amsthm}
\usepackage{graphicx}
\usepackage{hyperref}
\usepackage{algorithm,algpseudocode}
\usepackage[shortlabels]{enumitem}
\usepackage{braket}

\newcommand{\Pset}[0]{\Delta_d^\downarrow}
\newcommand{\K}[0]{\mathcal{K}}
\newcommand{\B}[0]{\mathcal{B}}
\newcommand{\C}[0]{\mathcal{C}}
\renewcommand{\P}[0]{\mathcal{P}}
\renewcommand{\S}[0]{\mathcal{S}}
\newcommand{\E}[0]{\mathcal{E}}
\newcommand{\pset}[0]{\Delta_d^{\downarrow}}
\newcommand{\xmax}[0]{\overline{x}^{(\epsilon)}}
\newcommand{\xmin}[0]{{\underline{x}}^{(\epsilon)}}
\newcommand{\nueps}[0]{\nu^{(\epsilon)}}
\newcommand{\postmaj}[1]{\succeq_{#1,\epsilon}}
\newcommand{\premaj}[1]{\,{}_{#1,\epsilon}\!\!\succeq}
\newcommand{\npostmaj}[1]{\nsucceq_{#1,\epsilon}}
\newcommand{\npremaj}[1]{\,{}_{#1,\epsilon}\!\!\nsucceq}
\newcommand{\maj}[0]{\succeq}
\newcommand{\nmaj}[0]{\not\succeq}
\newcommand{\eg}[0]{see\textit{, e.g.}, }
\renewcommand{\tr}[0]{\mathrm{Tr}}
\newcommand{\postapp}[1]{\overset{#1,\epsilon}{\mapsto}}
\newcommand{\preapp}[1]{\underset{#1,\epsilon}{\mapsto}}

\theoremstyle{plain}
  \newtheorem{theorem}{Theorem}[section]
  \newtheorem{lemma}[theorem]{Lemma}
  
  \newtheorem{corollary}[theorem]{Corollary}
\theoremstyle{definition}
  \newtheorem{definition}[theorem]{Definition}
  \newtheorem{example}[theorem]{Example}

\usepackage{xcolor}

\begin{document}

\title[Extremal elements]{Extremal elements of a sublattice of the majorization lattice and approximate majorization}

\author{C Massri$^{1,2}$, G Bellomo$^3$, F Holik$^{2,4}$, G M Bosyk$^{2,4,5}$}

\address{$^1$ Instituto de Investigaciones Matemáticas ``Luis A. Santaló'', UBA, CONICET, CABA, Argentina}
\address{$^2$ Grupo de Matem\'atica Aplicada, Departameto de Matem\'atica, Universidad CAECE, CABA, Argentina}
\address{$^3$ CONICET-Universidad de Buenos Aires, Instituto de Investigación en Ciencias de la Computación (ICC), Buenos
Aires, Argentina}
\address{$^4$ Instituto de F\'{\i}sica La Plata, UNLP, CONICET, Facultad de Ciencias Exactas, La Plata, Argentina}
\address{$^5$ Universit\`{a} degli Studi di Cagliari, Via Is Mirrionis 1, I-09123, Cagliari, Italy}

\ead{gbosyk@fisica.unlp.edu.ar}

\begin{abstract}
Given a probability vector $x$ with its components sorted in non-increasing order, we consider the closed ball
$\B^p_\epsilon(x)$ with $p \geq 1$ formed by the probability vectors whose $\ell^p$-norm distance to the center $x$ is less than or equal to a radius $\epsilon$. Here, we provide an order-theoretic characterization of these balls by using the majorization partial order.
Unlike the case $p=1$ previously discussed in the literature, we find that the extremal probability vectors, in general, do not exist for the closed balls $\B^p_\epsilon(x)$ with $1<p<\infty$.
On the other hand, we show that $\B^\infty_\epsilon(x)$ is a complete sublattice of the majorization lattice. 
As a consequence, this ball has also extremal elements. 
In addition, we give an explicit characterization of those extremal elements in terms of the radius and the center of the ball. 
This allows us to introduce some notions of approximate majorization and discuss its relation with previous results of approximate majorization given in terms of the $\ell^1$-norm. 
Finally, we apply our results to the problem of approximate conversion of resources within the framework of quantum resource theory of nonuniformity.
\smallskip

\end{abstract}

\noindent{\it Keywords}: majorization lattice, sublattice, approximate majorization

\maketitle

\section{Introduction}

Majorization has become a powerful mathematical tool with applications in different disciplines from economy to physics (\eg Ref.~\cite{MarshallBook} for an introduction to majorization and some of its applications). This is due to the fact that majorization provides an intuitive way of comparing probability vectors. 
In particular, majorization has also been widespread in the quantum realm (\eg~Refs.~\cite{Nielsen2001,Bellomo2019} and references therein for applications of majorization in quantum information). 
Indeed, in many resource theories, the deterministic and exact transformations between resources by means of free operations are governed by a majorization arrow between  probability vectors associated to the resources (\eg Ref.~\cite{Chitambar2019} for an  introduction to quantum resource theories and Refs.~\cite{Coecke2019,Fritz2019} for an abstract formulation of this formalism in terms of a symmetric monoidal category). 
For instance, this is the case of transformations of bipartite entangled pure states (resources) by means of local operations and classical communication (the free operations), as it is stated by the celebrated Nielsen's theorem~\cite{Nielsen1999}, where the associated probability vectors are formed using the Schmidt coefficients of the corresponding states. 
Other quantum resource theories where resource transformations are also given by a majorization arrow, are the ones of quantum coherence~\cite{Du2015,Chitambar2016,Zhu2017,Du2017} and nonuniformity (or purity)~\cite{Gour2015,Streltsov2018}.

As expected, it is not always possible to transform resources in an exact and deterministic way. 
For that reason, some other alternatives are studied. 
In particular, if one allows some degree of error, approximate transformations arise as the natural ones. 
For instance, approximate transformations within entanglement~\cite{Vidal2000,Owari2008}, coherence~\cite{Du2015} and thermodynamics~\cite{Chitambar2016,Chubb2018} quantum resource theories have already been studied.
In this context, the notion of approximate majorization in terms of the $\ell^1$-norm~\cite{Torgersen1991,Chubb2018,Horodecki2018} is one of the most considered in the literature. 
However, in most of these works the fully order-theoretic properties of majorization had not been considered.
Here, we aim to exploit these properties and subsequently analyze other possible definitions of approximate majorization.

Precisely, from the order-theoretic viewpoint, majorization is a binary relation that partially orders the set of probability vectors whose entries are sorted in a non-increasing manner (\eg \cite[Sec. I.B]{MarshallBook}). Moreover, it has been shown that there is a proper unique greatest lower bound (infimum) and a unique least upper bound (supremum) for any pair of probability vectors. This leads to the definition of the majorization lattice~\cite{Bapat1991,Cicalese2002}. This lattice is also complete, that is, the supremum and infimum exist for arbitrary subsets of probability vectors~\cite{Bapat1991,Bondar1994,Bosyk2019}. 
We remark that the lattice structure of majorization, beyond its partial order, has recently been found useful for different quantum information problems, namely the study of majorization uncertainty relations~\cite{Partovi2011,Wang2019,Li2019,Yuan2019}, entanglement transformations~\cite{Bosyk2017,Bosyk2018,Guo2018,Sauerwein2018}, and optimal common resource in majorization-based quantum resources theories~\cite{Bosyk2019}, among others~\cite{Hanson2018,Partovi2009,Korzekwa2017,Yu2019,Xiao2019,Hanson2019}.   

Here, we investigate the order-theoretic properties of a particular subset of probability vectors. 
More precisely, we consider the set $\B^p_\epsilon(x)$ (with $p\geq 1$) formed by probability vectors (with entries sorted in a non-increasing order) that are at $\ell^p$-norm distance $\epsilon$ from a given probability vector $x$. 
The case $p=1$ (i.e., considering the $\ell^1$-norm) is particularly interesting, since it can be related to notions of approximate majorization previously introduced in the literature~\cite{Torgersen1991,Chubb2018}. 
Indeed, this case has been recently considered in Refs.~\cite{Horodecki2018,Hanson2018}, where it is shown that the set $\B^1_\epsilon(x)$ admits extremal probability vectors, in the sense of majorization. 
We show that a similar result arises for the case $p=\infty$, that is, the set $\B^\infty_\epsilon(x)$ has also extremal probability vectors, leading to other notions of approximate majorization in terms of the $\ell^\infty$-norm. 
Moreover, we show that this is no longer true in general for other $\ell^p$-norms, turning the cases $p=1$ and $p=\infty$ particularly interesting. In order to prove our results, we first show that $\B^\infty_\epsilon(x)$ is a complete sublattice of the majorization lattice. This guarantees the existence of minimum and maximum elements. Next, we provide a characterization of them in terms of the radius and the center, and we relate them to two different notions of approximate majorization in terms of the $\ell^\infty$-norm distance.
Finally, we show that our results can be applied to the problem of approximate state conversion within the resource theory of nonuniformity \cite{Gour2015,Streltsov2018}.

The rest of this work is organized as follows. 
In Section~\ref{sec:majorization}, we recall the notions of majorization and majorization lattice, and its main properties. 
In Section~\ref{sec:bola}, we provide one of our main results, namely that the ball $\B^\infty_\epsilon(x)$ is a complete sublattice of the majorization lattice, and we give the explicit forms of the maximium an minimum elements.
In Section~\ref{sec:approx}, we introduce the notion of approximate majorization in terms of the $\ell^\infty$-norm distance and present its properties and discuss its relation with previous notions introduced in the literature. 
In Section~\ref{sec:nonuniform}, we apply our results to the framework of the quantum resource theory of nonuniformity.
Finally, some final remarks are drawn in Section~\ref{sec:conclusions}. 
For the sake of readability and completeness, all order-theoretic notions used in the main text are introduced in \ref{app:lattice}, whereas all technical details and proofs are delegated to Appendices B, C and D.

\section{Preliminaries: majorization lattice}
\label{sec:majorization}

Here, we recall some basics of majorization theory and the majorization lattice that will be useful later. 
For a more complete introduction to the subject \eg  Ref.~\cite{MarshallBook}.

We are going to work with probability vectors, which lie on a $d$-dimensional space: the $(d-1)$-probability simplex
\begin{equation} 
\label{eq:simplex}
  \Delta_d :=\left\{x \in \mathbb{R}^d \,\colon\, x_i\geq 0 \; \text{and} \; \sum_{i=1}^d x_i=1 \right\}.
\end{equation}
%
%
In this space, we can talk about majorization between two probability vectors (see, e.g. \cite{MarshallBook}) in the following sense.

\begin{definition}
\label{def:majorization}
For given $x,y \in \Delta_d$, it is said that $x$ majorizes $y$ (denoted by $x \succeq y$) if
\begin{equation} \label{eq:partialsums}
  s_k(x^\downarrow) \geq s_k(y^\downarrow), \quad \forall k \in \{1, \ldots, d-1\},
\end{equation}
where $x^\downarrow$ and $y^\downarrow$ denotes	 that the components of $x$ and $y$ are sorted in non-increasing order and $s_k(x):= \sum^k_{i=1} x_i$.
\end{definition}

Since we are working with probability vectors, we trivially have $s_d(x) = s_d(y) =1$ and, for that reason, we exclude this condition from the definition of majorization.
Definition~\ref{def:majorization} provides a natural way to see if one probability vector is more concentrated than another one.
Indeed, any probability vector $x \in \Delta_d$ trivially satisfies the relations $e_1 \succeq x  \succeq e_d$, with $e_1 := (1, 0, \dots, 0)$ and $e_d := \left(\frac{1}{d}, \ldots, \frac{1}{d} \right)$.

Remarkably, the majorization relation can be posed in several alternative ways.
A particular one, that relates majorization with doubly stochastic matrices, was originally discussed in the seminal work \cite{Hardy1929}. 
Precisely, a $d\times d$ matrix $B$ is doubly stochastic matrix if $B_{ij} \geq 0$ and $\sum_{i}B_{ij}=\sum_{j}B_{ij}= 1$ for all $i,j$. Therefore,
\begin{equation}
\label{eq:majdouble}
x \maj y \iff \text{there exists a doubly stochastic matrix such that } y = Bx.
\end{equation}

Moreover, by using Birkhoff’s theorem~\cite{Birkhoff1946}, which states that the set of  $d\times d$ doubly stochastic matrices coincides with the convex hull of the set of $d \times d$ permutation matrices, the r.h.s of~\eqref{eq:majdouble} turns out to be equivalent to 
$y= (\sum_k p_k \Pi_k)x$ for some $p \in \Delta_{d'}$ with $d'\leq d^2- 2d+2$ (see Ref.~\cite{Marcus1959}) and some set of permutation matrices $\{\Pi_k\}$.
In addition, notice that $e_d=Be_d$ for any doubly stochastic matrix $B$.
In this sense, majorization can be interpreted as a quantification of the notion of nonuniformity.

Among several equivalent definitions of majorization, a particularly useful one, for our purposes, appeals to the notion of Lorenz curve \cite{Lorenz1905}. More precisely, for a given probability vector $x$, one introduces the set of points $\left\{\left(k,s_k(x^\downarrow)\right)\right\}_{k=0}^d$ (with the convention $(0,0)$ for $k=0$). Let $L_x(\omega)$, with $\omega \in [0,d]$, be the polygonal curve obtained by the linear interpolation of these points. 
For $x \in \Delta_d$ this corresponds to the Lorenz curve of $x$, which is a non-decreasing and concave polygonal curve from $(0,0)$ to $(d,1)$ (\eg Fig. \ref{fig:figure1}.a and \ref{fig:figure1}.b). 
In this way, given two Lorenz curves $L_x(\omega)$ and $L_y(\omega)$, $L_x(\omega) \geq L_y(\omega)$ for all $\omega \in [0,d]$ implies that $x  \succeq y$, and vice versa.

Now, we recall some interesting order-theoretic properties of majorization. For the sake of completeness, all the order-theoretic notions that will be used here are defined in~\ref{app:lattice}. 
First, let us introduce the set of $d$-dimensional probability vectors whose components are sorted in non-increasing order,
\begin{equation} 
\label{eq:pset}
\pset :=\left\{x \in \Delta_d \,\colon\, x_i\ge x_{i+1}  \; \text{with} \; i \in \{1, \ldots, d-1\} \right\}.
\end{equation}
Notice that this set can be geometrically visualized as a convex polytope embedded in the $(d-1)$-probability simplex $\Delta_d$ (see, for instance, Figure~\ref{fig:figure1}.c, where the $2$-simplex $\Delta_3$ and $\Delta^\downarrow_3$ are depicted).

It can be shown that the set $\Pset$ equipped with the majorization relation $\maj$ given in Definition~\ref{def:majorization} is a \textit{partially ordered set} (POSET). 
Notice that if we relax the constraint that the components are sorted in a non-increasing manner, then antisymmetry condition is no longer valid in general (see \ref{app:lattice}). 
Instead, a weaker version holds, where $x$ and $y$ only differ by a permutation of its entries. 
In such case, the set $\Delta_d$ equipped with majorization relation gives a pre-ordered set, since the conditions of reflexivity and transitivity remain valid. 
On the other hand, $\langle \pset, \maj \rangle$ is not a totally ordered set, in general. 
This is due to the fact that there always exist $x,y \in \Pset$ such that $x \nmaj y$ and $y \nmaj x$ for any $d >2$. 
In this situation, we say that the probability vectors are incomparable. 
In terms of Lorenz curves, this means that they are different but intersect at least at one point in the interval $(1,d)$.
However, in such case, one can easily realize that there are infinite Lorenz curves below the ones of $x$ and $y$, and among of all them, there is one which is the greatest one. In the same vein, there are infinitely many Lorenz curves above those of $x$ and $y$, and there is one which is the lowest one (\eg Fig. \ref{fig:figure1}.a and \ref{fig:figure1}.b, where Lorenz curves of $x=(0.7, 0.2, 0.1)$, $y=(0.6, 0.35, 0.05)$, as well as, the corresponding ones of $x \vee y$ and $x \wedge y$ are depicted).

\begin{figure}
	\label{fig:figure1}
	\includegraphics[width=\textwidth]{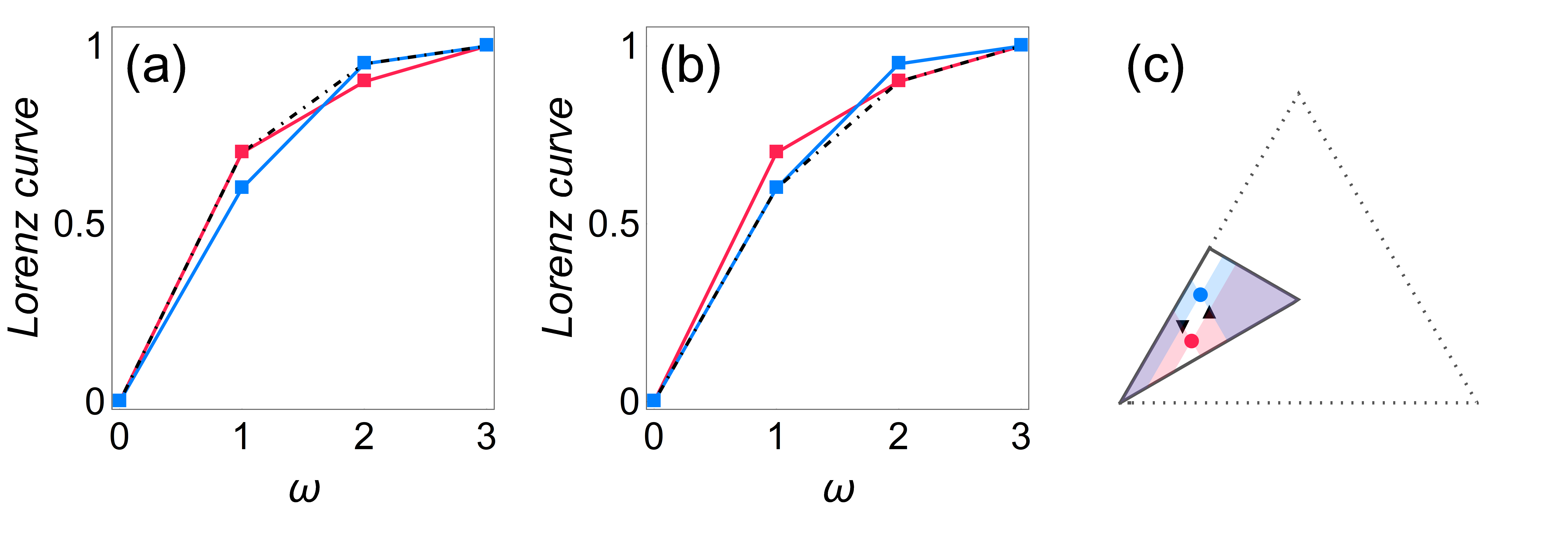}
	\caption{Let $x=(0.7, 0.2, 0.1)$ and $y = (0.6, 0.35, 0.05)$, then $x \vee y = (0.7, 0.25, 0.05)$ and $x \wedge y = (0.6, 0.3, 0.1)$ (a) Lorenz curves of $x$ (red), $y$ (blue) and $x \vee y$ (black dot-dashed) (b)  Lorenz curves of $x$ (red), $y$ (blue) and $x \wedge y$ (black dot-dashed) (c) 2-Simplex $\Delta_3$ (region inside the dotted triangle), $\Delta^\downarrow_3$ (region inside the black triangle), $x$ (red point), $y$ (blue point), $x \vee y$ ($\blacktriangledown$) and $x \wedge y$ ($\blacktriangle$). Blue region indicates the set $\left\{x' \in \Delta^\downarrow_3: x \maj x' \ \text{or} \ x' \maj x\right\}$, the red region indicates the set $\left\{x' \in \Delta^\downarrow_3: y \maj x' \ \text{or} \ x' \maj y\right\}$ and the violet region indicates the intersection of both sets.}
\end{figure}

These intuitions can be formalized and allow to formulate the notion of infimum, $x \wedge y$, and supremum, $x \vee y$, between two probability vectors $x,y \in \Pset$, that lead to the definition of the \emph{majorization lattice}~\cite{Bapat1991,Cicalese2002}. 
In particular, the algorithms to obtain $x \vee y$ and $x \wedge y$ were first introduced in \cite{Cicalese2002}.
Clearly, the majorization lattice is a bounded lattice, with top and bottom elements $e_1$ and $e_d$, respectively.
Moreover, it turns out that the majorization lattice is indeed \emph{complete}~\cite{Bapat1991,Bondar1994,Bosyk2019}.
In other words, the infimum and supremum exist for every family of probability vectors in $\Pset$. 
We reproduce this result below and the algorithms to obtain the corresponding infimum and supremum, as it will be useful for the rest of the work.
\begin{lemma}[see, for example, Prop. 1 of \cite{Bosyk2019}] \label{lemma:completnessmaj}
The POSET $\langle \Delta_d, \maj \rangle$ is a complete lattice, that is, for arbitrary $\mathcal{P}\subseteq \Pset$ there exist the infimum $x^{\inf} := \bigwedge \mathcal{P}$ and the supremum  $x^{\sup} := \bigvee \mathcal{P}$ of $\mathcal{P}$. 
The components of $x^{\inf}$ are given by
\begin{equation} \label{eq:infuncountableset}
  x^{\inf}_k = \inf  \mathcal{S}_k  - \inf \mathcal{S}_{k-1} ,
\end{equation}
where $\mathcal{S}_k = \{s_k(x): x \in \mathcal{P} \}$ with $s_0(x) := 0$. To obtain the components of $x^{\sup}$, we must first define the probability vector with components
\begin{equation} \label{eq:supuncountableset1}
  \bar{x}_k = \sup  \mathcal{S}_{k} - \sup \mathcal{S}_{k-1}.
\end{equation}
Then, we compute the upper envelope of $\bar{L}(\omega)$, the
polygonal curve defined by the linear interpolation of the points
$\{(k,s_k(\bar{x})) \}_{k=0}^d$.
Finally, the components of the supremum are given by
\begin{equation} \label{eq:supuncountableset}
  x^{\sup}_k =  \bar{L}(k) - \bar{L}(k-1).
\end{equation}
\end{lemma}

\noindent%
For completeness, in~\ref{app:upperenv} we recall the explicit algorithm to compute the upper envelope given in Ref.~\cite{Bosyk2019}.

When the set $\mathcal{P} \subseteq \Pset$ is a convex polytope, the corresponding infimum and supremum can be computed as the infimum and supremum of the set of vertices, $\mathrm{vert}(\mathcal{P})$, as explained in the following Lemma.
\begin{lemma}[see Lemma 1 of~\cite{Bosyk2019}] \label{lemma:convexplytope}
Let $\mathcal{P}$ be a convex polytope contained in $\Pset$, and $\mathrm{vert}(\mathcal{P})$ the set of vertices, $\mathrm{vert}(\mathcal{P}) = \{v^n\}^N_{n=1}$. Then, the infimum $x^{\inf} := \bigwedge \mathcal{P}$ and the supremum $x^{\sup} := \bigvee \mathcal{P}$ of $\mathcal{P}$ are given by the infimum and supremum elements of $\mathrm{vert}(\mathcal{P})$, namely
\begin{equation}
  x^{\inf} = \bigwedge \{ v^n \}^N_{n=1} \quad \text{and} \quad
  x^{\sup} = \bigvee \{ v^n\}^N_{n=1}.
\end{equation}
\end{lemma}

\noindent%
Clearly, the infimum and the supremum do not necessarily belong to the convex polytope. In the sequel, we restrict our attention to a
particular class of convex polytopes that admit extremal probability vectors and inherits the lattice structure of majorization.

\section{Order-theoretic properties of \texorpdfstring{$\B^\infty_{\epsilon}(x)$}{infinity-norm ball}}
\label{sec:bola}
Let $\B_\epsilon^p(x)$ be the (closed) ball with center $x \in\Delta_d^{\downarrow}$ and radius $\epsilon > 0$ inside $\pset$, that is,
\begin{equation} \label{eq:ball}
  \mathcal{B}^p_{\epsilon}(x) = \left\{ x' \in\Delta_d^{\downarrow}\,\colon\, \| x' - x\|_p \le \epsilon \right\},
\end{equation}
where $\|x\|_p = \left(\sum^d_{i=1} |x_i|^p\right)^{\frac{1}{p}}$ with $p \geq 1$, and a limiting case $\|x\|_\infty := \lim_{p \rightarrow \infty} \|x\|_p = \max\left\{|x_i|\right\}^d_{i=1}$. 
Here, we are interested in characterizing the order-theoretic properties of these balls with respect to the majorization relation.
In particular, we aim to find whether these balls admit extremal probabilities in the sense of majorization.

Our first result is that the balls $\B^p_\epsilon(x)$ with $1<p<\infty$ do not have extremal probabilities in general.
In others words, there exist probability vectors $x \in \Pset$, for which one can always find a radius $\epsilon$ sufficiently small such that the maximum or minimum of $\B_\epsilon^p(x)$ do not exist. More precisely, we obtain the following result.

\begin{theorem}
	\label{th:nominmax}	
	Let $x\in\Delta_d^{\downarrow}$, $\epsilon>0$ and $1<p<\infty$ such that $\B_{\epsilon}^p(x) \cap \partial \Pset = \emptyset$. Then, the supremum and infimum of the ball $\B_{\epsilon}^p(x)$, which are, $\bigvee \B_{\epsilon}^p(x)$ and $\bigwedge \B_{\epsilon}^p(x)$, do not belong to the ball. In other words, the maximum and the minimum of $\B_{\epsilon}^p(x)$ do not exist.
\end{theorem}

Interesting enough, this is not necessarily true for the closed $\ell^p$-norm balls with $p=1$ and $p=\infty$. 
We now provide a step forward in the order-theoretic characterization of these balls.
Let us first note that for
$p=1$ and $p=\infty$, $\B^p_\epsilon(x)$ is a convex polytope, 
since it is a translation and scaling of the convex polytope $\{x' \in \mathbb{R}^d: \|x'\|_p \leq 1 \}$ intersected with the $(d-1)$-simplex (\eg \cite{RockafellarBook}). 
Then, according to Lemma~\ref{lemma:convexplytope} the supremum  and infimum can be obtained from the vertices, $ \{ v^n \}^N_{n=1}$, of $\B^p_\epsilon(x)$, that is, $\bigvee \B^p_\epsilon(x) = \bigvee \{ v^n \}^N_{n=1}$  and $\bigwedge \B^p_\epsilon(x) = \bigwedge  \{ v^n \}^N_{n=1}$. 
This characterization is not useful enough to obtain the supremum and infimum, since one needs to know \textit{a priori} the vertices of $\B^p_\epsilon(x)$. 
Notwithstanding, for $p=1$, it has recently been shown not only how to compute supremum and infimum, but also that they are the maximum and minimum of the ball $\B^1_\epsilon(x)$~\cite{Horodecki2018,Hanson2018}.
Here, we complete the order-theoretic characterization of the balls by analyzing the case $p=\infty$, showing that $\B^\infty_\epsilon(x)$ also admits extremal probability vectors.
These results indicate that the existence of the maximum and minimum strongly depends on the geometry of the ball under consideration.

Hereafter, we focus on the $\ell^\infty$-norm $\epsilon$-ball with center $x$.
To show the existence of $\max \B^\infty_\epsilon(x)$ and $\min \B^\infty_\epsilon(x)$, without appealing to the vertices of $\B^\infty_\epsilon(x)$, we provide the following theorem that states $\B^\infty_\epsilon(x)$ inherits the lattice structure of majorization. Precisely, it is a complete sublattice of the majorization lattice.
\begin{theorem} \label{th:sublattice}
Let $x \in \pset$ and $\epsilon>0$. Then, the quadruple $\left\langle \B^\infty_\epsilon(x), \succeq, \vee, \wedge \right\rangle$ is a complete sublattice of the majorization lattice $\langle \Pset, \succeq, \vee, \wedge \rangle$.
\end{theorem}
This is an interesting order-theoretic result in itself that means for any $\B \subseteq  \B^\infty_\epsilon(x)$, one has that the supremum and the infimum of $\B$ belong to the ball, that is,
$\bigvee \B \in \B_\epsilon^\infty(x)$ and $\bigwedge \B \in \B_\epsilon^\infty(x)$, respectively.
As a consequence, the supremum and the infimum of $\B^\infty_\epsilon(x)$ belong to the ball, which means $\bigvee \B_\epsilon^\infty(x)$ and $\bigwedge \B_\epsilon^\infty(x)$ are indeed the maximum ($\max \B_\epsilon^\infty(x)$) and the minimum ($\min \B_\epsilon^\infty(x)$) of $\B^\infty_\epsilon(x)$, respectively. 
The following two theorems provide the algorithms to compute the maximum and minimum of $\B_\epsilon^\infty(x)$, respectively.

\begin{theorem}[Maximum] \label{th:max}
Let $x\in\Delta_d^{\downarrow}$, $\epsilon>0$ and $\xmax_\infty$
be a probability vector defined as
\[
\xmax_\infty:= x + \nu,
\]
where
\begin{equation} \label{eq:numax}
  \nu:= (\overbrace{\epsilon,\ldots,\epsilon}^{k_0-1},\epsilon-\delta, \overbrace{-\epsilon,\ldots,-\epsilon,}^{k_2-k_0} \overbrace{-x_{k_2+1},\ldots,-x_{d}}^{d-k_2}),
\end{equation}
$k_2 := \max \{k_0,k_1\}$, $\delta:=  k_0\epsilon - f(k_0)$, $1\leq k_0,k_1,k_2 \leq d$ and
\begin{align*}
  k_1 &:=\max \mathcal{K}_1, \quad \text{with} \;\; \mathcal{K}_1:=\{k\,\colon\, x_k\ge \epsilon\}\cup\{1\}, \label{eq:k1} \\
  k_0 &:=\min\mathcal{K}_0, \quad \text{with} \;\; \mathcal{K}_0:=\{k\,\colon\,f(k)\le k\epsilon\},\\
  f(k)&:=
    \begin{cases}
      1-s_k(x), & \text{if}\;k\ge k_1\\
      1-s_{k_1}(x)+(k_1-k)\epsilon, & \text{if}\;k\le k_1
    \end{cases}
\end{align*}
Then,
\begin{equation}
\max \B^\infty_{\epsilon}(x) = \xmax_\infty.
\end{equation}
\end{theorem}

\begin{theorem}[Minimum] \label{th:min}
Let $x\in\Delta_d^{\downarrow}$, $\epsilon>0$ and $\xmin_\infty$ be a probability vector defined as
\begin{equation} \label{eq:min_elem}
  {\xmin_\infty}_k :=  \bar{L}(k) - \bar{L}(k-1),
\end{equation}
where $\bar{L}(\omega)$ is the upper envelope of the polygonal curve given by the linear interpolation of the points $\{(k, s_k(x - \nueps)) \}_{k=0}^d$ with
\begin{equation} \label{eq:nueps}
  \nueps :=
   \begin{cases}
   (\epsilon,\ldots,\epsilon,-\epsilon,\ldots,-\epsilon)&\text{ if }d\text{ even},\\
   (\epsilon,\ldots,\epsilon,0,-\epsilon,\ldots,-\epsilon)&\text{ if }d\text{ odd}.\\
   \end{cases}
\end{equation}
Then,
\begin{equation}
  \min \B^\infty_\epsilon(x) = \xmin_\infty.
\end{equation}
\end{theorem}

Several simple cases can be noted. For example, if $\|e_1 - x\|_{\infty} \leq \epsilon$, then the vector~\eqref{eq:numax} reduces to $\nu = (1-x_1,-x_2, \ldots, -x_d)$ and $\xmax_\infty = e_1$, as expected. On the other hand, when $x + \nueps \in \B^\infty_\epsilon(x)$, with $\nueps$ given in~\eqref{eq:nueps}, one has $\nu = \nueps$, hence $\xmax_\infty = x + \nueps$. For the minimum, if $\|e_d - x\|_{\infty} \leq \epsilon$ then the construction in equation~\eqref{eq:min_elem} directly gives $\xmin_\infty = e_d$, as expected. Another simplification occurs when $x - \nueps \in \B^\infty_\epsilon(x)$. In this case, $\xmin_\infty = x - \nueps$. 
Given that the ball $\B^\infty_\epsilon(x)$ is defined by the intersection with the ordered simplex $\Delta^\downarrow_d$, there is always a large enough $\epsilon$ such that $x\pm \nu^{(\epsilon)}$ does not belong to the ball.

\medskip

As direct consequences of Theorem~\ref{th:min}, we can prove the following two corollaries. The first one, states that the minimal distributions of two balls with the same radius but centers that are related by a majorization relation, preserve such majorization relation. The second one, says that the minimum of a ball of a given radius can be obtained iterating the construction with two radius that sum to the original one.

\begin{corollary} \label{cor:preserveorder}
Let $x,y \in \Pset$ be two centers such that $x \succeq y$ and $\epsilon > 0$. Then,
\begin{equation}
  \xmin_\infty \succeq \underline{y}^{(\epsilon)}_\infty.
\end{equation}
\end{corollary}

\begin{corollary} \label{cor:semigroup}
Let $x\in \Pset$ and two radius $\epsilon_1, \epsilon_2 > 0$. Then,
\begin{equation} \label{eq:semigroup}
  \underline{x}^{(\epsilon_1+\epsilon_2)}_\infty = \underline{\left(\underline{x}^{(\epsilon_2)}_\infty\right)}^{(\epsilon_1)}_\infty.
\end{equation}
\end{corollary}

Unfortunately, in general, the above properties does not hold for the maximum. Let us see some counterexamples.
\begin{example}
Let $x=(0.5,0.3,0.2,.0)$, $y=(0.5,0.3,0.1,0.1)$ and $\epsilon=0.1$. Notice that $x \succeq y$, but $\xmax_\infty=(0.6,0.3,0.1,0) \nsucceq \overline{y}^{(\epsilon)}_\infty=(0.6,0.4,0,0)$.
\end{example}

\begin{example}
Let $x=(0.5,0.2,0.15,0.1,0.05)$, $\epsilon_1=0.1$ and $\epsilon_2=0.05$. Then, $\overline{x}^{(\epsilon_1+\epsilon_2)}_\infty=(0.65,0.35, 0, 0, 0) \neq \overline{\left(\overline{x}^{(\epsilon_2)}_\infty\right)}^{(\epsilon_1)}_\infty = (0.65, 0.3, 0.05, 0, 0)$, with $\overline{x}^{(\epsilon_2)}_\infty = (0.55, 0.25, 0.15, 0.05, 0)$.
\end{example}

\subsection{Special cases: \texorpdfstring{$d=3$}{d=3} and admissible pairs}
Here, we address two special cases. In the first example, we consider three-dimensional probability vectors ($d=3$), whereas in the second, we restrict to certain values of the pair $(x,\epsilon)$.

For $d=3$, notice that all the results obtained by using the $\ell^\infty$-norm, can be directly translated to results using the $\ell^1$-norm, and vice versa, since we have the following equivalence between the corresponding balls.

\begin{theorem}
	\label{prop:equivballs}
	Let $x\in\Delta_3^\downarrow$ and $\epsilon>0$. Then,
	\begin{equation} \label{eq:equivalenceballs}
	\B^1_{\epsilon}(x) = \B^\infty_{\frac{\epsilon}{2}}(x).
	\end{equation}
\end{theorem}

Now, we restrict our consideration to pairs $(x,\epsilon)$ satisfying the following definition.

\begin{definition} \label{def:admissible}
	Let $x\in\Pset$ and $\epsilon >0$. A pair $(x,\epsilon)$ satisfying the following conditions is called an \emph{admissible pair},
	\begin{enumerate}
		\item $1>x_1>x_2>\ldots>x_{d}>0$.
		\item $0<2\epsilon\le x_i-x_{i+1}$ for all $i \in \{0, \ldots,d \}$,
	\end{enumerate}
	where $x_0:=1$ and $x_{d+1}:=0$.
\end{definition}

\noindent%
From a geometrical point of view, this means that 
$\B_{\epsilon}^\infty(x) \cap \partial \Pset = \emptyset$,
that is to say that an admissible pair always defines a ball that is small enough to be completely contained within the ordered simplex $\Pset$. Notice that, although this is not the most general case, it is representative of the class of approximation problems. In fact, in the scenario where one looks for an approximation, $x$ will be typically not in the border of the ordered simplex, and $\epsilon$ is expected to be small enough such that the ball is strictly contained within the same set. In this case, the maximum and minimum are easily calculated by just adding or subtracting the vector $\nueps$.

\begin{corollary} \label{cor:addmisible_minmax}
	Let $(x, \epsilon)$ be an admissible pair. Then,
	\begin{equation} \label{eq:maxminapair}
	\xmax_\infty = x + \nueps \quad \text{and} \quad \xmin_\infty = x - \nueps.
	\end{equation}
\end{corollary}

Admissible pairs also guarantees that the maximum satisfies similar properties to those given in Corollaries~\ref{cor:preserveorder} and \ref{cor:semigroup}.

\begin{corollary} \label{cor:preserveordermax}
	Let $(x, \epsilon)$ and $(y,\epsilon)$ be two admissible pairs such that $x \succeq y$. Then,
	\begin{equation}
	\xmax_\infty \succeq \overline{y}^{(\epsilon)}_\infty.
	\end{equation}
\end{corollary}

\begin{corollary} \label{cor:semigroupmax}
	Let $x\in \Pset$ and two radius $\epsilon_1, \epsilon_2 > 0$. Then,
	\begin{equation} \label{eq:semigroupmax}
	\overline{x}^{(\epsilon_1+\epsilon_2)}_\infty = \overline{\left(\overline{x}^{(\epsilon_2)}_\infty\right)}^{(\epsilon_1)}_\infty.
	\end{equation}
\end{corollary}


\section{Approximate majorization in terms of the \texorpdfstring{$\ell^\infty$-norm}{infinity-norm}}
\label{sec:approx}
Let us move on to the case of probability vectors $x,y\in\Delta_d^{\downarrow}$ such that $x\nsucceq y$. 
On one hand, one can look for a  probability vector $y'$ up to a certain distance of $y$, so that the majorization relation $x \maj y'$ is obtained. 
On the other hand, one can look for a probability vector $x'$ up to a certain distance  of $x$, so that the majorization relation $x' \maj y$ is obtained.
Let us now formally introduce these two notions of approximate majorization in terms of a distance $D$, which is a function $D: \pset \times \pset \mapsto [0,+\infty)$ satisfying the axioms of a metric: 
\begin{itemize}
	\item $D(x,y) \geq 0$ (non-negativity),
	\item $D(x,y) = 0 \iff x=y$ (identity of indiscernibles), 
	\item $D(x,y)=D(y,x)$ (symmetry), and 
	\item $D(x,y) \leq D(x,z)+D(y,z)$ (triangle inequality).
\end{itemize}	
Precisely,
\begin{definition} \label{def:postpremaj}
Let $x,y, \in \Pset$ and $\epsilon \geq 0$. We say that:
\begin{itemize}
  \item $x$ $\epsilon$-post-majorizes $y$, denoted by $x \postmaj{D} y$, whenever there exist $y' \in \Pset$ such that $x \succeq y'$ and $D(y',y) \leq \epsilon$, and
  \item $x$ $\epsilon$-pre-majorizes $y$, denoted by $x \premaj{D} y$, whenever there exist $x' \in \Pset$ such that $x' \succeq x$ and $D(x',x)  \leq \epsilon$.
\end{itemize}
\end{definition}
Clearly, these definitions depend on the choice of the distance and there is no, in principle, any reason to choice one distance over other. 
In the following, we assume the distance $D_p(x,y) = \| x -y \|_p$ induced by the $\ell^p$-norm with $p \geq 1$. 
As it is well-known, all the $\ell^p$-norms are equivalent (\eg \cite{NariciBook}). Indeed, for any $p <q$, one has \cite{Gonzalez2015} 
\begin{equation}
	\|x\|_q \leq \|x\|_p \leq d^{\frac{1}{p} - \frac{1}{q}} \|x\|_q.
\end{equation}
This equivalence means that all $\ell^p$-norms define the same topology and the results of continuity, convergence and many other properties do not depend on which norm is chosen. In particular, we have this freedom of choice in the notions of post and pre-majorization given in Def.~\ref{def:postpremaj}. 
Notwithstanding, we will see that, for a practical purpose, the cases $p=1$ and $p=\infty$ become particularly  useful, since they can be related to the maximum and minimum of the corresponding balls. 

Let us first  observe that the relevant case $p=1$, where the corresponding induced distance $D_1(x,y)= \|x-y\|_1$ has a clear operational interpretation in terms of a degree of distinguishability between probability vectors, has already been considered in the literature~\cite{Torgersen1991,Chubb2018,Horodecki2018}. 
More precisely, let define $\overline{x}^{(\epsilon)}_1 := \max \B^1_\epsilon(x)$ and $\underline{y}^{(\epsilon)}_1:= \min \B^1_\epsilon(y)$. It has been shown that~\cite{Torgersen1991,Chubb2018,Horodecki2018} 
\begin{align}
\label{eq:postnorm1}
x \postmaj{D_1} y  &\iff   x \maj \underline{y}^{(\epsilon)}_1, \mbox{and} \\ 
x \premaj{D_1}  y &\iff \overline{x}^{(\epsilon)}_1 \maj y. \label{eq:prenorm1}
\end{align}
Moreover, post and pre-majorization by using the $\ell^1$-norm are indeed equivalent: $x \postmaj{D_1} y \iff x \premaj{D_1}  y$~\cite{Chubb2018}.

Let us now consider the case of $p=\infty$, where the corresponding induced distance is given by $D_\infty(x,y)= \|x-y\|_\infty$. 
Our first result is that post and pre-majorization in terms of $\ell^\infty$-norm have indeed an analogue characterization to the corresponding ones given in Eqs.~\eqref{eq:postnorm1}--~\eqref{eq:prenorm1}.

%
\begin{theorem}
	\label{th:approxmajequiv}
Let $x,y, \in \Pset$ and $\epsilon \geq 0$. Then,
\begin{align}
\label{eq:postnorminf}
x \postmaj{D_\infty} y  &\iff   x \maj \underline{y}^{(\epsilon)}_\infty, \mbox{and} \\ 
x \premaj{D_\infty}  y &\iff \overline{x}^{(\epsilon)}_\infty \maj y. \label{eq:prenorminf}
\end{align}
%
\end{theorem}
In this way, post and pre-majorization in terms of $\ell^1$-norm and $\ell^\infty$-norm are both equally useful definitions. 
Moreover, we will see that there is no a general criterion in terms of majorization to choice one distance over the other (see Tables~\ref{tab:postmaj} and~\ref{tab:premaj}). 
Therefore, the choice of the distance depends on the particular problem considered.  

As we recalled above, post and pre-majorization in terms of the $\ell^1$-norm are equivalent. 
We will see that this is not true in general for the case of $\ell^\infty$-norm. 
To show it, we provide the following examples.
\begin{example}
Let $x = \left(\frac{7}{13}, \frac{4}{13}, \frac{2}{13}, 0\right)$ and let $y = \left(\frac{4}{7}, \frac{3}{7}, 0, 0\right)$. Notice that letting $\epsilon = \frac{1}{10}$ we have,
\[%
\textstyle{%
\xmax_\infty = \left(\frac{83}{130}, \frac{40}{130}, \frac{7}{130}, 0\right) \quad \text{and} \quad \underline{y}^{(\epsilon)}_\infty = \left(\frac{33}{70}, \frac{23}{70}, \frac{1}{10}, \frac{1}{10}\right).
}\]
Hence, $x \succeq \underline{y}^{(\epsilon)}_\infty$ (or, equivalently $x \postmaj{D_\infty} y$), but $\xmax_\infty \nsucceq y$ (or, equivalently $x \npremaj{D_\infty} y$).

Now, let $x = \left(\frac{1}{3}, \frac{1}{3}, \frac{1}{3}\right)$ and let $y = \left(\frac{3}{5}, \frac{2}{5}, 0\right)$. 
In this case, for $\epsilon = \frac{3}{10}$ we have,
\[%
\textstyle{%
\xmax_\infty = \left(\frac{19}{30}, \frac{10}{30}, \frac{10}{30}\right) \quad \text{and} \quad \underline{y}^{(\epsilon)}_\infty = \left(\frac{7}{20}, \frac{7}{20}, \frac{3}{10}\right).
}\]
Hence,  $\xmax_\infty\succeq y$ (or, equivalently $x \premaj{D_\infty} y$), but $x\nsucceq \underline{y}^{(\epsilon)}_\infty$ (or, equivalently $x \npostmaj{D_\infty} y$).
\end{example}

Now, we ask for the minimal $\epsilon$ such that $x \postmaj{D_\infty} y$ is possible. 
We find that this quantity is equivalent to the minimal $\ell^\infty$-distance between $y$ and all probability vectors majorized by $x$.
In addition, we find a sharp upper bound to this quantity in terms of the $\ell^\infty$-distance between $x$ and $y$.

\begin{theorem}
	\label{prop:uppboundspostmaj}
Let $x,y\in\Delta_d^{\downarrow}$ be such that $x\nsucceq y$ and define $\epsilon_\infty,\tilde{\epsilon}_\infty$ as
\begin{equation}
  \epsilon_\infty:=\min\{\epsilon\,\colon\,x\postmaj{D_\infty} y\} \ \text{and} \ \tilde{\epsilon}_\infty:=\min\{\|y'-y\|_{\infty}\,\colon\,x\succeq y'\}.
\end{equation}
Then, $\epsilon_\infty = \tilde{\epsilon}_\infty \le \|x-y\|_\infty$, where the bound is sharp.
\end{theorem}

On the other hand, we obtain similar results for the case of $\epsilon$-pre-majorization.

\begin{theorem}
	\label{prop:uppboundspremaj}
Let $x,y\in\Delta_d^{\downarrow}$ be such that $x\nmaj y$ and define $\epsilon^\infty,\tilde{\epsilon}^\infty$ as
\begin{equation}
  \epsilon^\infty := \min\{\epsilon\,\colon\,x \premaj{D_\infty} y\} \quad \text{and} \quad \tilde{\epsilon}^\infty:=\min\{\|x'-x\|_{\infty}\,\colon\,x'\succeq y\}.
\end{equation}
Then, $\epsilon^\infty=\tilde{\epsilon}^\infty \leq \|x-y\|_\infty$, where the bound is sharp.
\end{theorem}

\hfill

In Table~\ref{tab:postmaj} and~\ref{tab:premaj}, we show with several examples that post and pre-majorization in terms of $\ell^\infty$-norm and $\ell^1$-norm are in equal foot in a majorization sense.
All the computations regarding $\ell^1$-norm were made by implementing the algorithms given in \cite{Horodecki2018}. 
We also use the following notation: $\epsilon_1:=\min\{\epsilon\,\colon\,x\postmaj{D_1} y\}$ and $\epsilon^1:=\min\{\epsilon\,\colon\,x\premaj{D_1} y\}$.

\begin{center}
\begin{table} \footnotesize
	\begin{tabular}{@{}lll@{}}
		\toprule[0.5mm]
		Example                                                  &  $\underline{y}^{(\epsilon_\infty)}_\infty$  and $\underline{y}^{(\epsilon_1)}_1$                                                  & Majorization comparison \\ \midrule[0.3mm]
		\begin{tabular}[c]{@{}l@{}}$x=(1/3, 4/15, 1/5, 1/5)$\\  $y=(4/9, 5/18, 5/18, 0)$\end{tabular} & \begin{tabular}[c]{@{}l@{}}$\underline{y}^{(1/5)}_\infty= (4/15, 4/15, 4/15, 1/5)$ \\
			$\underline{y}^{(2/5)}_1=(4/15, 4/15, 4/15, 1/5)$\end{tabular} & $\underline{y}^{(\epsilon_\infty)}_\infty=\underline{y}^{(\epsilon_1)}_1$             \\ \midrule
		\begin{tabular}[c]{@{}l@{}}$x=(4/11, 7/22, 2/11, 3/22)$\\ $y=(1/2, 1/3, 1/6, 0)$\end{tabular}  & \begin{tabular}[c]{@{}l@{}}$\underline{y}^{(1/5)}_\infty= (3/10, 1/4, 1/4, 1/5)$ \\ $\underline{y}^{(10/33)}_1=(23/66, 1/3, 1/6, 5/33)$\end{tabular} & $\underline{y}^{(\epsilon_\infty)}_\infty\not\succeq\underline{y}^{(\epsilon_1)}_1$ and
		$\underline{y}^{(\epsilon_1)}_1\succeq\underline{y}^{(\epsilon_\infty)}_\infty$                 \\ \midrule
		\begin{tabular}[c]{@{}l@{}}$x=(6/13, 3/13, 2/13, 2/13)$\\$y=(4/11, 4/11, 3/11, 0)$\end{tabular} & \begin{tabular}[c]{@{}l@{}}$\underline{y}^{(1/10)}_\infty= (3/10, 3/10, 3/10, 1/10)$ \\ $\underline{y}^{(72/143)}_1=(1/4, 1/4, 1/4, 1/4)$\end{tabular} & $\underline{y}^{(\epsilon_\infty)}_\infty\succeq\underline{y}^{(\epsilon_1)}_1$ and
		$\underline{y}^{(\epsilon_1)}_1\not\succeq\underline{y}^{(\epsilon_\infty)}_\infty$         \\ \midrule
		\begin{tabular}[c]{@{}l@{}}	$x=(4/11, 7/22, 2/11, 3/22)$\\$y=(1/2, 1/3, 1/6, 0)$\end{tabular} & \begin{tabular}[c]{@{}l@{}}	$\underline{y}^{(1/5)}_\infty= (3/10, 1/4, 1/4, 1/5)$ \\$\underline{y}^{(10/33)}_1=(23/66, 1/3, 1/6, 5/33)$\end{tabular} & $\underline{y}^{(\epsilon_\infty)}_\infty\not\succeq\underline{y}^{(\epsilon_1)}_1$ and
		$\underline{y}^{(\epsilon_1)}_1\not\succeq\underline{y}^{(\epsilon_\infty)}_\infty$
		\\ \bottomrule[0.5mm]
	\end{tabular}
\caption{Comparison between post-majorization in terms of $\ell^\infty$-norm and $\ell^1$-norm.}
\label{tab:postmaj}
\end{table}
\end{center}

\begin{center}
	\begin{table}  \footnotesize
		\begin{tabular}{@{}lll@{}}
			\toprule[0.5mm]
			Example                                                  &  $\overline{x}^{(\epsilon^\infty)}_\infty$ and $\overline{x}^{(\epsilon^1)}_1$   & Majorization comparison \\ \midrule[0.3mm]
			\begin{tabular}[c]{@{}l@{}}$x=(4/13, 4/13, 3/13, 2/13)$\\$y=(1, 0, 0, 0)$\end{tabular} & \begin{tabular}[c]{@{}l@{}}$\overline{x}^{(7/10)}_\infty= (1, 0, 0, 0)$\\
				$\overline{x}^{(18/13)}_1=(1, 0, 0, 0)$\end{tabular} & $\overline{x}^{(\epsilon^\infty)}_\infty = \overline{x}^{(\epsilon^1)}_1$             \\ \midrule
			\begin{tabular}[c]{@{}l@{}}$x=(9/29, 8/29, 7/29, 5/29)$\\$y=(1/2, 3/7, 1/14, 0)$\end{tabular}  & \begin{tabular}[c]{@{}l@{}}$\overline{x}^{(1/5)}_\infty= (74/145, 13/29, 6/145, 0)$\\
			$\overline{x}^{(139/203)}_1=(202/203, 1/203, 0, 0)$ \end{tabular} & $\overline{x}^{(\epsilon^\infty)}_\infty\not\succeq\overline{x}^{(\epsilon^1)}_1$ and
			$\overline{x}^{(\epsilon^1)}_1\succeq\overline{x}^{(\epsilon^\infty)}_\infty$                 \\ \midrule
			\begin{tabular}[c]{@{}l@{}}	$x=(1/3, 7/24, 7/24, 1/12)$\\$y=(4/11, 3/11, 3/11, 1/11)$\end{tabular} & \begin{tabular}[c]{@{}l@{}}$\overline{x}^{(1/10)}_\infty= (13/30, 3/8, 23/120, 0)$\\
			$\overline{x}^{(5/66)}_1=(9/22, 7/24, 7/24, 1/132)$\end{tabular} & 	$\overline{x}^{(\epsilon^\infty)}_\infty\succeq\overline{x}^{(\epsilon^1)}_1$ and
			$\overline{x}^{(\epsilon^1)}_1\not\succeq\overline{x}^{(\epsilon^\infty)}_\infty$         \\ \midrule
			\begin{tabular}[c]{@{}l@{}}	$x=(8/25, 7/25, 7/25, 3/25)$\\$y=(7/18, 5/18, 2/9, 1/9)$\end{tabular} & \begin{tabular}[c]{@{}l@{}}	$\overline{x}^{(1/10)}_\infty= (21/50, 19/50, 9/50, 1/50)$\\$\overline{x}^{(31/225)}_1=(103/225, 7/25, 59/225, 0)$\end{tabular} & 	$\overline{x}^{(\epsilon^\infty)}_\infty\not\succeq\overline{x}^{(\epsilon^1)}_1$ and
			$\overline{x}^{(\epsilon^1)}_1\not\succeq\overline{x}^{(\epsilon^\infty)}_\infty$
			\\ \bottomrule[0.5mm]
		\end{tabular}
		\caption{Comparison between pre-majorization in terms of $\ell^\infty$-norm and $\ell^1$-norm.}
		\label{tab:premaj}
	\end{table}
\end{center}

\normalsize

\section{Approximate state transformations between nonuniform states}
\label{sec:nonuniform}
There are many resource theories where the deterministic and exact transformations between resources are governed by a majorization law. That is the case for the so-called quantum nonuniformity resource theory, which claims that any nonuniform (that is, non maximally mixed) state is a resource, and studies the possible interconversion between resources by means of a prescribed class of (free) operations~\cite{Gour2015}. This problem is closely related to the thermodynamic scenario where the free states and operations are defined as the ones that are thermal relative to some fixed temperature~\cite{Janzing2000}.

Our previous results can be applied to the resource theory of quantum nonuniformity. For simplicity, we follow Ref.~\cite{Streltsov2018}, where only quantum operations that preserve the dimension of the Hilbert space are considered (a more general case is discussed in~\cite{Gour2015}).

Regarding the free operations, which must preserve the set of free states (in our case, the maximally mixed state $I/d$), there are different approaches in the literature. Namely,
\begin{itemize}
	\item Mixture of unitaries (MU): $\E_{\mathrm{MU}}(\rho) = \sum p_i U \rho U_i^\dag$, with $p \in \Delta_{d'}$ (for some $d'$) and $\{U_i\}$ unitary operations;
	\item Noisy operations (NO): $\E_{\mathrm{NO}}(\rho) = \tr_E\left( U (\rho \otimes I/d)U^\dag\right)$, with $U$ a unitary operation;
	\item Unital operations (U): $\E_{\mathrm{U}}(I/d) = I/d$.	
\end{itemize}
Notice that the induced sets of free operations satisfy the strict inclusion relations $\{\E_{\mathrm{MU}} \} \subset \{\E_{\mathrm{NO}} \}\subset \{\E_{\mathrm{U}}\}$ (see Lemma 5 of~\cite{Gour2015}). Regardless of which class of free operations is considered, the set of free states is formed only by the uniform state $\rho = I/d$ (i.e. the maximally mixed one). Hence, the resources are called nonuniform states.

Moreover, as regards the transformations between nonuniform states, the three classes of quantum free operations are equivalent, as the following Lemma asserts.
\begin{lemma}[Lemma 10 of~\cite{Gour2015}]
	Let $\rho$ and $\sigma$ be two quantum density matrices acting on a $d$-dimensional Hilbert space.
	Then, 
	\begin{equation}
    \rho \underset{MU}{\mapsto} \sigma \iff \rho \underset{NO}{\mapsto} \sigma \iff \rho \underset{U}{\mapsto} \sigma,
	\end{equation}
	%
where $\rho \underset{X}{\mapsto} \sigma$ with $X \in \{MU,NO,U\}$ means that there exists a free operation $\E_X$ such that $\sigma=\E_X(\rho)$.
\end{lemma}

In this way, one can fix any of these three classes of free operations without loss of generality, and let $\rho \mapsto \sigma$ denote the transformation of $\rho$ into $\sigma$ by means of a free operation (no matter of which class). In addition, due to Uhlmann's~\cite{Uhlmann1970} theorem, the problem of exact transformations between nonuniform states reduces to a majorization relation between their corresponding spectra (\eg Lemma 3 of~\cite{Streltsov2018}). Precisely,
\begin{lemma}
	\label{lemma:nonunifmaj}
Let $\rho$ and $\sigma$ be two quantum density matrices acting on a $d$-dimensional Hilbert space.
	Then, 
\begin{equation}
\rho \mapsto \sigma \iff x(\rho) \maj x(\sigma),
\end{equation}
where $x(\rho),\,x(\sigma) \in \Pset$, with their components given by the eigenvalues of $\rho$ and $\sigma$, respectively.
\end{lemma}

Clearly, there exist nonuniform states $\rho$ and $\sigma$ such that $\rho \mapsto \sigma$ is not allowed. In this case, free operations are not enough to give a solution for the exact transformation. An alternative route is to relax the exact condition. For example, one can look for a final state $\sigma'$ up to a certain distance of $\sigma$, so that the transformation $\rho \mapsto \sigma$ is allowed. Or one can look for an initial state $\rho'$ up to a certain distance of $\rho$, so that $\rho' \mapsto \sigma$ is allowed. Let us define some notation for those cases.
\begin{definition} \label{def:approxtransf}
	Let $\rho$ and $\sigma$ be two quantum states acting on  a $d$-dimensional Hilbert space. We write:
	\begin{itemize}
		\item $\rho \postapp{D} \sigma$ whenever there exists a quantum state $\sigma'$ such that $\rho \mapsto \sigma'$ and $D(\sigma,\sigma') \leq \epsilon$,
		\item $\rho \preapp{D} \sigma$ whenever there exists a quantum state $\rho'$ such that $\rho' \mapsto \sigma$ and $D(\rho,\rho') \leq \epsilon$,
	\end{itemize}
	where $D$ is a distance (a metric) on the set of quantum states.
\end{definition}
These can be seen as notions of quantum post and pre-majorization. As it is noticed in Ref.~\cite{Gour2015}, two natural questions arise: 
(a) which distance measure on the set of quantum states should be used?; and (b) can the problem be reduced to a majorization relation between probability vectors associated to the quantum states, as in the case of exact transformations (see Lemma~\ref{lemma:nonunifmaj})?
For the latter, 
the answer is positive as long as the distance measure used is contractive under unital operations. This contractivity means that, for any pair of quantum states $\rho$ and $\sigma$, and for every unital operation $\E$, then $D(\rho,\sigma) \ge D(\E(\rho),\E(\sigma))$. It turns out that this condition is satisfied by any distance induced by the Schatten $p$-norm of quantum states, $D_p(\rho,\sigma):=\|\rho -\sigma\|_p = (\tr|\rho -\sigma|^p)^{1/p}$ with $p \geq 1$ (see \cite{Perez2006}).
\begin{theorem} \label{th:equivapproxtransf}
Let $\rho$ and $\sigma$ be two quantum states acting on  a $d$-dimensional Hilbert space and $x(\rho), x(\sigma) \in \Pset$ the corresponding probability vectors formed by the eigenvalues of $\rho$ and $\sigma$ sorted in non-increasing order, respectively. 
Then, statement \ref{sti} is equivalent to \ref{stii} 
and statement \ref{stiii} is equivalent to \ref{stiv},

\begin{enumerate}[(i)]
\item\label{sti} $\rho \postapp{D_p} \sigma$
\item\label{stii} $x(\rho)  \postmaj{D_p} x(\sigma)$
\item\label{stiii} $\rho \preapp{D_p} \sigma$
\item\label{stiv} $x(\rho) \premaj{D_p} x(\sigma)$
\end{enumerate}
\end{theorem}

In other words, Theorem~\ref{th:equivapproxtransf} says that, in the resource theory of nonuniformity, the approximate conversion of quantum states by means of free operations can be studied in terms of classical probability distributions (we notice that the equivalence between \ref{sti} and \ref{stii} is similar to Lemma 55 of~\cite{Gour2015}, but here we relax the contractivity condition of the metric to free operations). 
Hence, all our previous results concerning approximate majorization can be applied.

\section{Concluding remarks}
\label{sec:conclusions}

In this work, we study the majorization relation in connection with the geometry given by the $\ell^p$-norm. 
More precisely, we prove that $\ell^p$-balls with $1<p <\infty$, in general, do not admit extremal probaility vectors (Theorem~\ref{th:nominmax}), unlike the case $p=1$ previously discussed in the literature.
Here, we complete this study by considering the case $p=\infty$.
In particular, we show that the balls $\B_\epsilon^\infty$ are complete sublattices of the majorization lattice (Theorem~\ref{th:sublattice})). 
This is an interesting order-theoretic result in itself that has as a consequence that each $\ell^\infty$-ball has supremum and infimum that belong to the ball, that is, they are the maximum (Theorem~\ref{th:max}) and minimum (Theorem~\ref{th:min}), respectively. 
Our findings show that the existence of a maximum and a minimum, is a rather peculiar feature, which is specific of the $\ell^\infty$ and $\ell^1$ norms. 

We stress that this property becomes of particular interest in those areas of research for which the notion of \emph{approximate majorization} becomes essential, due to the impossibility of working without making approximations. 
In turn, our results suggest that the case $\ell^\infty$-norm (Theorem~\ref{th:approxmajequiv}) together with $\ell^1$-norm are the most relevant ones for practical applications where approximations are unavoidable. 
And in principle, as it follows from the examples given in Tables~\ref{tab:postmaj} and~\ref{tab:premaj},
there is no criterion in terms of majorization to choose one norm over the other.

Finally, we apply our results in the framework of the quantum resource theory of nonuniformity.
In particular, we obtain that post-majorzation and  pre-majorzation between quantum states is equivalent to its classical counterpart between the probability vectors formed by the corresponding eigenvalues of the quantum sates (equivalences between (i) and (ii),  and (iii) and (iv) of Theorem \ref{th:equivapproxtransf}, respectively). 
In this way, all our previous results concerning approximate majorization can be applied to this resource theory.

\section*{Acknowledgements}

All the authors acknowledge financial support from Consejo Nacional de Investigaciones Cient\'ificas y T\'ecnicas (AR) and Agencia Nacional de Promoci\'on Cient\'ifica y Tecnol\'ogica (AR). 
G.M.B is partially supported by the Fondazione di Sardegna within the project ``Strategies and Technologies for Scientific Education and Dissemination''.

\appendix

\section{A brief introduction to order theory}
\label{app:lattice}

Here, we provide all order-theoretic definitions that we have used in the main text to characterize the majorization lattice and the ball $\B^\infty_{\epsilon}(x)$ (\eg \cite{DaveyBook}, for a complete introduction to order and lattice theories).

One of the primitive concepts in order theory is the one of \emph{partial order}, which formalizes and generalizes the intuitive notion of ordering of the elements of a set. 
\begin{definition}[POSET]
A (non-strict) partial order is a binary relation, denoted as $\maj$, over a set $\mathcal{P}$ satisfying the following axioms: 
for every $x,y,z \in \mathcal{P}$ one has
\begin{enumerate}
	\item reflexivity: $x \succeq x$,
	\item antisymmetry: $x \succeq y$ and $y \succeq x$, then $x=y$, and
	\item transitivity: $x \succeq y$ and $y \succeq z$, then $x \succeq  z$.
\end{enumerate}
\end{definition}
Notice that if the antisymmetry condition is not satisfied, the binary relation $\maj$ is called a \textit{preorder}.

A set equipped with a partial order, $\langle \P, \maj \rangle$, is called a \emph{partially ordered set} (POSET).
A \textit{totally ordered} set is a particular POSET where all elements are comparable, that is: either $x \maj y$ or $y \maj x$ for any pair $x,y \in \P$. 
Clearly, this is not the general case, because there exist POSETs for which there are elements that are incomparable, that is, $x \nmaj y$ and $y \nmaj x$.

The \textit{top} (or maximal) and \textit{bottom} (or minimal) elements of a POSET, if any, are defined as follows.
\begin{definition}[top and bottom]
	A POSET $\langle \P, \maj \rangle$ has  
	\begin{enumerate}
		\item a top element $\top$ iff $\top \in \P$ and $\top \maj x$ for all $x \in \S$ 
		\item a bottom element $\bot$ iff $\bot \in \P$ and $x \maj \bot$ for all $x \in \S$
	\end{enumerate}
\end{definition}
A \textit{bounded} POSET is a POSET that has top and bottom elements. 
Other important notions in order theory are the \emph{supremum} (or least upper bound) and the \emph{infimum} (or greatest lower bound) of a given subset.
\begin{definition}[supremum and infimum]
Let $\langle \P, \maj \rangle$ be a POSET and let $\S \subset \P$. 
On the one hand, the supremum of $\S$ (if exists), denoted as $\bigvee \S$, is defined as an element of $\P$ satisfying the two following conditions:
\begin{enumerate}
	\item $\bigvee \S$ is an upper bound: $\bigvee \S \maj x$ for all $x \in \S$ 
	\item $\bigvee \S$ is the least of the upper bounds: for each $y \in \P$ such that $y \maj x$ for all $x \in \S$, one has $\bigvee \S \maj y$.
\end{enumerate}
On the other hand, the infimum of $\S$ (if exists), denoted as $\bigwedge \S$, is defined as an element of $\P$ satisfying the two conditions:
\begin{enumerate}
	\item $\bigwedge \S$ is a lower bound: $x \maj \bigwedge \S$ for all $x \in \S$ 
	\item $\bigwedge \S$ is the greatest of the lower bounds: for each $y \in \P$ such that $x \maj y$ for all $x \in \S$, one has $y  \maj \bigwedge \S$.
\end{enumerate}
\end{definition}
Notice that if the supremum exists, then it is unique (and the same happens for the infimum). In addition, if $\bigvee \S \in \S$, then it is called the maximum of $\S$ and denoted as $\max\S$. In similar way,  if $\bigwedge \S \in \S$, then it is called the minimum of $\S$ and denoted as $\min\S$.

A \emph{lattice} is a partially ordered set in which the supremum and infimum of any two elements exist.
\begin{definition}[lattice]
A lattice $\langle \P, \maj, \vee, \wedge \rangle$ is a POSET $\langle \P, \maj\rangle$ such that for any pair $x,y \in \P$ there exist the supremum, denoted as $x \vee y$, and the infimum, denoted as $x \wedge y$. 
\end{definition}
This is the order-theoretic definition of a lattice.
An alternative but equivalent definition is the algebraic one, where a lattice is defined as a set equipped with two binary operators,
$\vee$ and $\wedge$, which satisfy the idempotent, commutative and associative laws, as well as the absorption law (\eg  \cite{DaveyBook}).
Accordingly, if one considers a subset of $\P$ formed by finite elements, $\S = \{ x^1,\ldots, x^N \}$ with $\ x^i \in \P$, by appealing to the algebraic properties of the definition of lattice, it is straightforward to show that the infimum and the supremum of $\P$ always exist, and are given by $\bigwedge \mathcal{P} = x^1 \wedge x^2 \wedge \ldots \wedge x^N$ and $\bigvee \mathcal{P} = x^1 \vee x^2 \vee \ldots \vee x^N$, respectively. However, if one considers an arbitrary subset of $\P$ (which could be infinite), the lattice properties alone are not strong enough to guarantee the existence of an infimum and a supremum. If infimum and supremum exist for every family, the lattice is said to be \emph{complete}. 
\begin{definition}[complete lattice]
	A complete lattice $\langle \P, \maj, \vee, \wedge \rangle$ is a POSET $\langle \P, \maj\rangle$ such that for any $\S \in \P$ there exist the supremum, denoted as $\bigvee \S$, and the infimum, denoted as $\bigwedge \S$. 
\end{definition}
Finally, we introduce the notion of a \emph{sublattice}, which is a subset with the same supremum and infimum of the original lattice.	
\begin{definition}[sublattice]
	Let $\langle \P, \maj, \vee, \wedge \rangle$ be a lattice. 
	$\S \subseteq \P$ is a subalttice  of $\P$ iff $x \vee y \in \S$ and $x \wedge y \in \S$ for any pair $x,y \in \S$.
\end{definition}

\section{Algorithm to calculate the upper envelope} 
\label{app:upperenv}

Here, we recall the algorithm given in Ref.~\cite{Bosyk2019} to obtain the  upper envelope of the polygonal joining $\{(k,s_k(x))\}_{k=0}^d$ for a given probability vector $x\in\mathbb{R}^d$.

\begin{algorithm}[H]
    \caption{Upper envelope}
    \label{alg:upperenv}
    \begin{algorithmic}
        \State \textbf{input:} $x\in\mathbb{R}^d$
        \State \textbf{output:} coordinates of the upper envelope of the polygonal joining $\{(k,s_k(x))\}_{k=0}^d$
        \smallskip
        \Procedure{UpperEnv}{$x$}
        \State $\K \leftarrow \{0\}$ \Comment{stores the `critical points' of $x$}
        \State $i \leftarrow 0$
            \While{$i < \mathrm{length}(x)$}
                \State $m \leftarrow \{0\}$ \Comment{stores slope values}
                \For{$j = i\!+\!1\,\ldots\,\mathrm{length}(x)$}
                    \State $m \leftarrow \mathrm{append}\left\{m,\frac{s_j(x)-s_i(x)}{j-i}\right\}$
                \EndFor
                \State $k \leftarrow$ $\max$(position of $\max(m)$) \Comment{finds position of the last maximum slope}
                \State $\K \leftarrow \mathrm{append}\{\K,k\}$
                \State $i \leftarrow k$ \Comment{updates $i$}
            \EndWhile
            \State \textbf{return} $\{(k,s_{k}(x))\}_{k \in \K}$ \Comment{coordinates of the upper envelope}
        \EndProcedure
    \end{algorithmic}
\end{algorithm}

\section{Proofs of Sec.~\ref{sec:bola}}
\label{app:appendixproofsbola}

\begin{proof}[{Proof of Theorem} \ref{th:nominmax}:]
	Let $x\in\mathbb{R}^d$, $\epsilon>0$, $1<p<\infty$ and consider the ball $B_\epsilon^p(x) = \{ x'\in\mathbb{R}^d \,\colon\, \|x'-x\|_p \leq \epsilon \}$. 
Do not confuse the notation $B_\epsilon^p(x)$ with $\B_{\epsilon}^p(x)=B_\epsilon^p(x)\cap \Pset$.
	Let $y,y'\in B_\epsilon^p(x)$ and $z(t):=ty+(1-t)y'$ for $t\in[0,1]$. Let us check that $\|z(t)-x\|_p<\epsilon$ for $t\in(0,1)$. Let define the function $f(t):= (\|z(t) - x\|_p)^p = \sum_{i=1}^d|ty_i+(1-t)y'_i-x_i|^p$. It is straightforward to show that $f(t)$ is convex, with exactly one local minimum and $f(0)=f(1)=\epsilon^p$. Then, $f(t)<\epsilon^p$ for all $t\in (0,1)$.
	Therefore, these calculations imply that any linear variety intersecting the ball $B_\epsilon^p(x)$ in more than one point, necessarily intersects the interior of the ball. Clearly, this does not hold for $p=1$ or $p=\infty$.
	
	Given $x\in\mathbb{R}^d$ we can define the sets
	\begin{align*}
	\overline{C}_x:=& \left\{x'\in\mathbb{R}^d\,\colon\, s_k(x')\leq s_k(x)\quad\forall k\in \{1,\ldots, d\} \right\},\\
	\underline{C}_x:=& \left\{x'\in\mathbb{R}^d\,\colon\, s_k(x')\geq s_k(x)\quad\forall k\in \{1\ldots, d\} \right\}
	\end{align*}
	Both $\overline{C}_x$ and $\underline{C}_x$ are convex (unbounded) polyhedra.
	
	Let now assume $x\in\Delta_d^{\downarrow}$, $\epsilon>0$ and $1<p<\infty$ such that $\B_{\epsilon}^p(x) \cap \partial \Pset = \emptyset$. Recall that the ball $\B_{\epsilon}^p(x)$ is equal to $B_\epsilon^p(x)\cap \Pset$ and let us denote its supremum by $\bigvee \B_\epsilon^p(x):=\overline{x}^{(\epsilon)}$. Notice that, being $\overline{x}^{(\epsilon)}$ the supremum of $\B_{\epsilon}^p(x)$, the polytope $\overline{\C}_{\overline{x}^{(\epsilon)}}:=\overline{C}_{\overline{x}^{(\epsilon)}}\cap \Pset$ must contain the ball $\B_{\epsilon}^p(x)$. Then, from our previous calculations, if a face of $\overline{\C}_{\overline{x}^{(\epsilon)}}$ intersects $\B_{\epsilon}^p(x)$, then this intersection consists of one point.
	
	Let assume that $\overline{x}^{(\epsilon)}\in \B_{\epsilon}^p(x)$ and let us make the following remarks:
	\begin{enumerate}
		\item \label{point1} Notice that, if $\overline{x}^{(\epsilon)}$ is in the interior of $\B_{\epsilon}^p(x)$, then there exists a face of $\overline{\C}_{\overline{x}^{(\epsilon)}}$ that intersects the ball in more than one point. Hence, $\|\overline{x}^{(\epsilon)}-x\|_p=\epsilon$.
		\item \label{point2} If a face of $\overline{\C}_{\overline{x}^{(\epsilon)}}$ intersects the ball $\B_{\epsilon}^p(x)$ in one point, then this point has to be the supremum. Therefore, this face is tangential to the $\B_{\epsilon}^p(x)$ at $\overline{x}^{(\epsilon)}$.
		\item \label{point3} From the equations of $\overline{C}_x$, it follows that the one-dimensional faces of $\overline{\C}_{\overline{x}^{(\epsilon)}}$ passing through $\overline{x}^{(\epsilon)}$ are given by $\gamma_1, \ldots, \gamma_{d-1}$, with
		\begin{equation}
		\gamma_i:=\{t(0,\ldots 0,1,-1,0,\ldots,0)+\overline{x}^{(\epsilon)}: \ t\ge 0\}
		\end{equation}
		where $t$ is bounded above by the condition $\gamma_i \subseteq \Delta_d^{\downarrow}$ and the entry $1$ of the $d$-dimensional vector $(0,\ldots 0,1,-1,0,\ldots,0)$ is in position $i$.
	\end{enumerate}
	
	Let us consider a maximal dimensional face of $\overline{\C}_{\overline{x}^{(\epsilon)}}$ that intersects the ball $\B_{\epsilon}^p(x)$. From items (\ref{point1}) and (\ref{point2}), this face has to be tangential to the ball at $\overline{x}^{(\epsilon)}$.  By the geometry of the ball, the vector $\overline{x}^{(\epsilon)}-x$ is normal to the tangent space at $\overline{x}^{(\epsilon)}$. Hence, $(\overline{x}^{(\epsilon)}-x)\bot\gamma_i$ for all $i=1,\ldots,d-1$. Then, by the characterization given in (\ref{point3}) it follows $\overline{x}^{(\epsilon)}=x+t(1,\ldots,1)\not\in\Pset$ for $t\neq 0$, a contradiction.
	
	A similar argument applies for the minimum.
\end{proof}

\hfill

\begin{proof}[{Proof of Theorem} \ref{th:sublattice}:]
	Let $x \in \pset$ and $\epsilon>0$, and consider an arbitrary subset $\B \subseteq \B^\infty_\epsilon(x)$. Let $\bar{s}_k:=\sup\{s_k(x')\,\colon x'\in\B\}$ with $k \in \left\{0, \ldots, d \right\}$, and let $\bar{L}(\omega)$ be the upper envelope of the piecewise linear curve interpolating the points $\left\{(k,\bar{s}_k)\right\}_{k=0}^d$. According to Lemma~\ref{lemma:completnessmaj}, $\bar{L}(\omega)$ is the Lorenz curve of the supremum of $\B$. Then, $\bigvee \B:= x^{\sup} = \left(\bar{L}(1), \bar{L}(2) - \bar{L}(1), \ldots, \bar{L}(d) - \bar{L}(d-1)\right)$. Thus, we must prove that $\| x^{\sup} - x\|_\infty \leq \epsilon$. First, it is important to remark that, by construction, the set of indices $\K:=\{k\,\colon\, \bar{L}(k)=\bar{s}_k\}$ is non-empty, with $0,1,d \in\K$. Notice that, given $k$, there exist $k_0,k_1\in\K$ such that $k_1 < k \le k_0$ and
	\begin{equation}
	\bar{L}(k)-\bar{L}(k-1)= \bar{s}_{k_0}-\bar{L}(k_0-1)=\bar{L}(k_1+1)-\bar{s}_{k_1}.
	\end{equation}
	Let us prove that $|x^{\sup}_k-x_k|\leq \epsilon$ for all $k \in \left\{1, \ldots, d\right\}$. Let us first note that
	\begin{align}
	x^{\sup}_k-x_k &= \bar{L}(k)-\bar{L}(k-1) - x_k  \label{eq:s1} \\
	& = \bar{s}_{k_0}-\bar{L}(k_0-1) - x_k  \label{eq:s2} \\
	&\leq \bar{s}_{k_0}- \bar{s}_{k_0-1} -x_{k_0}   \label{eq:s3} \\
	&= \sup\{s_{k_0}(x'): x' \in \B \}-\sup\{s_{k_0-1}(x') : x' \in \B\} - x_{k_0} \label{eq:s4} \\
	& \leq \sup \{s_{k_0}(x')-s_{k_0-1}(x') : x' \in B \}-x_{k_0}  \label{eq:s5} \\
	&= \sup\{x'_{k_0}: x' \in \B \} - x_{k_0} \le \epsilon, \label{eq:s6}
	\end{align}
	where we have used $\bar{s}_{k_0-1} \leq \bar{L}(k_0-1)$ and $x_{k_0} \leq x_k$ from~\eqref{eq:s2} to~\eqref{eq:s3}, and $\sup \mathcal{A} - \sup \mathcal{B} \leq \sup \{\mathcal{A} -\mathcal{B}\}$ for $\mathcal{A},\mathcal{B} \subseteq \mathbb{R}$ from~\eqref{eq:s4} to~\eqref{eq:s5}, respectively.
	
	On the other hand, for a given $\delta>0$, there exists $x'$ such that $\bar{s}_{k_1}-\delta< s_{k_1}(x')$. Then,
	\begin{align}
	x^{\sup}_k -x_k + \delta &=\bar{L}(k)-\bar{L}(k-1) - x_k + \delta \label{eq:ss1}\\
	&= \bar{L}(k_1+1)-\bar{s}_{k_1} - x_k + \delta \label{eq:ss2} \\
	&\geq  \bar{s}_{k_1+1}-\bar{s}_{k_1} -x_{k_1+1}+ \delta \label{eq:ss3}\\
	&> \bar{s}_{k_1+1} - s_{k_1}(x')-x_{k_1+1} \label{eq:ss4} \\
	&\geq s_{k_1+1}(x')-s_{k_1}(x')-x_{k_1+1} \label{eq:ss5} \\
	& =x'_{k_1+1}-x_{k_1+1} \label{eq:ss6}\\
	&\geq -\epsilon.
	\end{align}
	where we have used $\bar{s}_{k_1+1}\leq \bar{L}(k_1+1)$ and $x_k \leq x_{k_1}$ from~\eqref{eq:ss2} to~\eqref{eq:ss3}. Then, $x^{\sup}_k -x_k + \delta >-\epsilon$ for all $\delta>0$, hence $x^{\sup}_k -x_k \ge -\epsilon$. Finally, $x^{\sup} \in \B^\infty_\epsilon(x)$ and $x^{\sup} \succeq x'$ for all $x'\in \B$.
	
	A more direct argument applies for the infimum.
\end{proof}

\hfill

\begin{proof}[{Proof of Theorem} \ref{th:max}:]
	According to Theorem~\ref{th:sublattice}, the supremum of $\B^\infty_\epsilon(x)$ is indeed the maximum. Let us denote it as $\max \B^\infty_\epsilon(x) := x + \nu^{\max}$, so that we have to prove that $\nu^{\max} = \nu$. On the one hand, by construction, one has that $\max\{-x_k, -\epsilon\}\leq \nu^{\max}_k \leq \epsilon$ for all $k \in \{1, \ldots, d\}$.
	On the other hand, notice that $x + \nu^{\max} \succeq x + \nu$ if and only if
	\begin{equation} \label{eq:sumaknumax2}
	s_k(\nu^{\max}) \geq s_k(\nu) \ \text{for all} \ k \in \{1, \ldots, d\}.
	\end{equation}
	Let us compare these inequalities by analyzing the following cases:
	\begin{itemize}
		\item if $k \in \{1, \ldots, k_0 -1 \}$, one has $s_k(\nu) = kr$. Hence, $s_k(\nu^{\max}) \leq s_k(\nu)$. This inequality together with~\eqref{eq:sumaknumax2} lead to $\nu^{\max}_k = \nu_k$;
		\item if $k \in \{k_0 +1, \ldots, d \}$, one has $\nu^{\max}_k \geq \nu_k=\max\{-x_k,-\epsilon\}$. Then, $\sum^d_{i=k} \nu_i^{\max} \geq \sum^d_{i=k} \nu_i$ and this is equivalent to $s_k(\nu^{\max}) \leq s_k(\nu)$. Therefore, the latter inequality together with~\eqref{eq:sumaknumax2} give $\nu^{\max}_k = \nu_k$;
		\item if $k = k_0$, one has $\nu^{\max}_{k_0} = \nu_{k_0}$, since $s_d(\nu^{\max})=s_d(\nu)$ and the two points given above.
	\end{itemize}
	Hence, $\nu^{\max} = \nu$. Therefore, $\max \B^\infty_\epsilon(x) = \xmax_\infty$.
\end{proof}

\hfill

\begin{proof}[{Proof of Theorem} \ref{th:min}:]
	Let us prove that $\xmin_\infty:=(\bar{L}(1),\bar{L}(2) - \bar{L}(1),\ldots,\bar{L}(d)-\bar{L}(d-1))$ satisfies $\xmin_\infty \in \B_\epsilon^\infty(x)$ and, for all $x'\in\B_\epsilon^\infty(x)$, one has $L_{x'}(\omega) \geq \bar{L}(\omega)$ for all $\omega \in [0,d]$. This is equivalent to prove that $\min \B^\infty_\epsilon(x) = \xmin_\infty$.
	
	Let us begin with the proof of $\xmin_\infty \in \B_\epsilon^\infty(x)$. Before that, let us observe that, by construction, the set of indices 
	$\K:=\{k\,\colon\, \bar{L}(k)=s_k(x - \nueps)\}$ is non-empty, with $0,d \in\K$. Notice that given $k$, there exists $k_0,k_1\in\K$ such that $k_1 < k \le k_0$ and
	\begin{equation}
	\bar{L}(k)-\bar{L}(k-1)= s_{k_0}(x - \nueps)-\bar{L}(k_0-1)=\bar{L}(k_1+1)- s_{k_1}(x - \nueps).
	\end{equation}
	Let us prove that $|\xmin_{\infty\,k}-x_k|\leq \epsilon$ for all $k \in \left\{1, \ldots, d\right\}$. Let us first note that
	\begin{align}
	\xmin_{\infty\,k}-x_k &= \bar{L}(k)-\bar{L}(k-1) - x_k  \label{eq:s1min} \\
	& = s_{k_0}(x - \nueps)-\bar{L}(k_0-1) - x_k  \label{eq:s2min} \\
	&\leq s_{k_0}(x - \nueps)- s_{k_0-1}(x - \nueps) -x_{k_0}   \label{eq:s3min} \\
	&= \nueps_{k_0} \leq \epsilon, \label{eq:s4min}
	\end{align}
	where we have used $\bar{s}_{k_0-1} \leq \bar{L}(k_0-1)$ and $x_{k_0} \leq x_k$ from~\eqref{eq:s2min} to~\eqref{eq:s3min}. On the other hand, let us note that
	\begin{align}
	\xmin_{\infty\,k}-x_k &= \bar{L}(k)-\bar{L}(k-1) - x_k  \label{eq:ss1min} \\
	& = \bar{L}(k_1+1)- s_{k_1}(x - \nueps) - x_k  \label{eq:ss2min} \\
	&\geq s_{k_1+1}(x - \nueps)- s_{k_1}(x - \nueps) -x_{k_1}   \label{eq:ss3min} \\
	&= \nueps_{k_1} \geq - \epsilon, \label{eq:ss4min}
	\end{align}
	where we have used $s_{k_1+1}(x - \nueps)\leq \bar{L}(k_1+1)$ and $x_k \leq x_{k_1}$ from~\eqref{eq:ss2min} to~\eqref{eq:ss3min}. 
	Therefore, $\xmin_\infty\in \B^\infty_\epsilon(x)$, and its Lorenz curve is given by $\bar{L}(\omega)$.
	
	Let us denote by $L_{x - \nueps}(\omega)$ to the polygonal curve given by the linear interpolation of the points $\{(k, s_k(x - \nueps)) \}_{k=0}^d$. Although $L_{x - \nueps}$ is not necessarily a Lorenz curve, by construction, it satisfies that for all $x' \in \B^\infty_\epsilon(x)$, one has $L_{x'}(\omega) \geq L_{x - \nueps}(\omega)$ for all $\omega \in [0,d]$. On the other hand, by appealing to the definition of upper envelope, one has that, for all $x' \in \B^\infty_\epsilon(x)$, $L_{x'}(\omega) \geq \bar{L}(\omega) \geq L_{x - \nueps}(\omega)$ for all $\omega \in [0,d]$.
\end{proof}

\hfill

\begin{proof}[{Proof of Corollary} \ref{cor:preserveorder}:]
	Let us introduce the polygonal curves $L_{x - \nueps}(\omega)$ and $L_{y - \nueps}(\omega)$ with $\omega \in  [0,d]$ given by the linear interpolation of the sets of points $\left\{\left(k,s_k(x - \nueps)\right)\right\}_{k=0}^d$ and $\left\{\left(k,s_k(y - \nueps)\right)\right\}_{k=0}^d$ with $\nueps$ given by~\eqref{eq:nueps}, respectively. Notice that these curves are not necessarily Lorenz curves. Given that $s_k(x) \geq s_k(y)$ for all $k \in \{1, \ldots, d\}$, then $L_{x - \nueps}(\omega)\geq L_{y - \nueps}(\omega)$ for all $\omega \in  [0,d]$. As a consequence, their respective upper envelopes preserve this order. Finally, by
	appealing to Theorem~\ref{th:min}, one has 
	$\xmin_\infty \succeq \underline{y}^{(\epsilon)}_\infty$.
\end{proof}

\hfill

\begin{proof}[{Proof of Corollary} \ref{cor:semigroup}:]
Let $L_{x - \nu^{(\epsilon_2)}}(\omega)$, $L_{\underline{x}^{(\epsilon_2)} - \nu^{(\epsilon_1)}}(\omega)$ and $L_{x - \nu^{(\epsilon_1+\epsilon_2)}}(\omega)$ the polygonal curves given by the linear interpolation of the points $\left\{\left(k,s_k(x - \nu^{(\epsilon_2)})\right)\right\}_{k=0}^d$, $\left\{\left(k,s_k(\underline{x}^{(\epsilon_2)} - \nu^{(\epsilon_1)})\right)\right\}_{k=0}^d$ and $\left\{\left(k,s_k(x - \nu^{(\epsilon_1+\epsilon_2)})\right)\right\}_{k=0}^d$, respectively; and let $\overline{L}_{x - \nu^{(\epsilon_2)}}(\omega)$, $\overline{L}_{\underline{x}^{(\epsilon_2)}-\nu^{(\epsilon_1)}}(\omega)$ and $\overline{L}_{x - \nu^{(\epsilon_1+\epsilon_2)}}(\omega)$ their corresponding upper envelopes.  Notice that~\eqref{eq:semigroup} is equivalent to $\overline{L}_{x - \nu^{(\epsilon_1+\epsilon_2)}}(\omega) = \overline{L}_{\underline{x}^{(\epsilon_2)}-\nu^{(\epsilon_1)}}(\omega)$
	for all $\omega \in [0,d]$.
	
	Let us focus our attention in an arbitrary interval $[k,k']$ where $L_{\underline{x}^{(\epsilon_2)}}(\omega)$ is linear. Notice that $L_{x - \nu^{(\epsilon_2)}}(\omega) = \overline{L}_{x - \nu^{(\epsilon_2)}}(\omega)$ for $\omega = k, k'$ and $L_{x - \nu^{(\epsilon_2)}}(\omega) \leq \overline{L}_{x - \nu^{(\epsilon_2)}}(\omega)$ for all $\omega \in (k,k')$. As a consequence, one has $L_{x - \nu^{(\epsilon_1+\epsilon_2)}}(\omega) = L_{\underline{x}^{(\epsilon_2)} - \nu^{(\epsilon_1)}}(\omega)$ for $\omega = k, k'$ and $L_{x - \nu^{(\epsilon_1+\epsilon_2)}}(\omega) \leq L_{\underline{x}^{(\epsilon_2)} - \nu^{(\epsilon_1)}}(\omega)$ for all $\omega \in (k,k')$, being $L_{\underline{x}^{(\epsilon_2)} - \nu^{(\epsilon_1)}}(\omega)$ a convex function in this interval. Therefore, their upper envelopes have to coincide in this interval, that is, $\overline{L}_{x - \nu^{(\epsilon_1+\epsilon_2)}}(\omega) = \overline{L}_{\underline{x}^{(\epsilon_2)}-\nu^{(\epsilon_1)}}(\omega)$ for all $\omega \in [k,k']$. Repeating this argument for all intervals of the form $[k,k']$ where $L_{\underline{x}^{(\epsilon_2)}}(\omega)$ is linear, one obtains the desired result $\overline{L}_{x - \nu^{(\epsilon_1+\epsilon_2)}}(\omega)= \overline{L}_{\underline{x}^{(\epsilon_2)}-\nu^{(\epsilon_1)}}(\omega)$ for all $\omega \in [0,d]$.
\end{proof}

\hfill

\begin{proof}[{Proof of Theorem} \ref{prop:equivballs}:]
	First, let us prove an intermediate result. Let $\mathcal{V}^p_\epsilon=\{\nu \in\mathbb{R}^3\,\colon\,\|\nu\|_p\le \epsilon \ \text{and} \ \nu_1+\nu_2+\nu_3=0\}$. Then, $\mathcal{V}^1_\epsilon = \mathcal{V}^\infty_{\frac{\epsilon}{2}}$. Indeed, notice that
	\begin{itemize}
		\item If $|\nu_1|\ge |\nu_2|$ and $\nu_1\ge 0\ge \nu_2$, then
		\[
		\|\nu\|_1=\nu_1-\nu_2+|\nu_3|=\nu_1-\nu_2+\nu_1+\nu_2=2\nu_1\le \epsilon \iff \|\nu\|_\infty = \nu_1 \le \frac{\epsilon}{2}.
		\]
		\item If $|\nu_1|\ge |\nu_2|$ and $\nu_2\ge 0\ge \nu_1$, then
		\[
		\|\nu\|_1=-\nu_1+\nu_2+|\nu_3|=-\nu_1+\nu_2-\nu_1-\nu_2=2|\nu_1|\le \epsilon \iff \|\nu\|_\infty = |\nu_1| \le \frac{\epsilon}{2}.
		\]
		\item If $|\nu_2|\ge |\nu_1|$ and $\nu_2\ge 0\ge \nu_1$, then
		\[
		\|\nu\|_1=-\nu_1+\nu_2+|\nu_3|=-\nu_1+\nu_2+\nu_1+\nu_2=2\nu_2\le \epsilon \iff \|\nu\|_\infty = \nu_2 \le \frac{\epsilon}{2}.
		\]
		\item If $|\nu_2|\ge |\nu_1|$ and $\nu_1\ge 0\ge \nu_2$, then
		\[
		\|\nu\|_1=\nu_1-\nu_2+|\nu_3|=\nu_1-\nu_2-\nu_1-\nu_2=2|\nu_2| \le \epsilon \iff \|\nu\|_\infty = |\nu_2| \le \frac{\epsilon}{2}.
		\]
	\end{itemize}
	
	Finally, let us notice that
	\[
	x' \in \B^1_\epsilon(x) \iff x'-x \in \mathcal{V}^1_\epsilon  = \mathcal{V}^\infty_{\frac{\epsilon}{2}} \iff x' \in\B^\infty_{\frac{\epsilon}{2}}(x).
	\]
\end{proof}

\hfill

\begin{proof}[{Proof of Corollary} \ref{cor:addmisible_minmax}:]
    Let $(x,\epsilon)$ be an admissible pair as in Definition~\ref{def:admissible}. Before presenting the proof, let us observe that any admissible pair satisfies $x_k \geq x_{k+1} + 2\epsilon \geq 2\epsilon$.
    
    Now, from Theorem~\ref{th:max} and using the previous observation, it is easy to see that $k_0 = d/2$ when $d$ is even, and $k_0 = (d+1)/2$ when $d$ is odd. Then, $k_1 = k_2 = d$, $\delta = \epsilon$ and $f(k) = (d-k)\epsilon$. Finally, using those parameters one has $\nu = \nueps$, arriving to the desired result: $\xmax = x + \nueps$.
    
    For the minimum, we invoke Theorem~\ref{th:min}. In this case, it is direct to observe that $x - \nueps \in \Pset$, that is to say that $x - \nueps$ is an ordered probability vector. Hence, $\bar{L}(\omega)$ coincides with the Lorenz curve of $x - \nueps$, and $\xmin_k = (x - \nueps)_k$, concluding the proof.
\end{proof}

\begin{proof}[{Proof of Corollary} \ref{cor:preserveordermax}:]
	From~\eqref{eq:maxminapair}, one has $\xmax_\infty = x+ \nueps$ and $\overline{y}^{(\epsilon)}_\infty=y+ \nueps$. Given that $x \succeq y$, it directly follows $\xmax_\infty = x+ \nueps \succeq \overline{y}^{(\epsilon)}_\infty=y+ \nueps$.
\end{proof}

\hfill

\begin{proof}[{Proof of Corollary} \ref{cor:semigroupmax}:]
	From ~\eqref{eq:maxminapair}, one has $\overline{x}^{(\epsilon_1+\epsilon_2)}_\infty
	= x + \nu^{(\epsilon_1+\epsilon_2)}$ and $\overline{x}^{(\epsilon_2)}_\infty= x + \nu^{(\epsilon_2)}$. Then, $\overline{\left(\overline{x}^{(\epsilon_2)}_\infty\right)}^{(\epsilon_1)}_\infty= \overline{(x + \nu^{(\epsilon_2)})}^{(\epsilon_1)}_\infty = x + \nu^{(\epsilon_2)} + \nu^{(\epsilon_1)} = x + \nu^{(\epsilon_1+\epsilon_2)} = \overline{x}^{(\epsilon_1+\epsilon_2)}_\infty$.
\end{proof}

\section{Proofs of Sec.~\ref{sec:approx}:}
\label{app:proofsapprox}

\begin{proof}[{Proof of Theorem} \ref{th:approxmajequiv}:]
	Let us first show the equivalence between $x \postmaj{D_\infty} y$ and $x \succeq \underline{y}^{(\epsilon)}_\infty$. Assume $x \postmaj{D_\infty} y$.
Then, there exists $y' \in \Pset$ such that $x \succeq y'$ and $\|y' - y\|_\infty \leq \epsilon$. Then, $y' \in \B_\epsilon^\infty(y)$ and, by definition, one has $y' \succeq \underline{y}^{(\epsilon)}_\infty$. Thus, if $x \postmaj{D_\infty} y \Rightarrow x\succeq \underline{y}^{(\epsilon)}_\infty$. The converse statement is straightforward, since by definition, one has $\|\underline{y}^{(\epsilon)}_\infty - y\|_\infty \leq \epsilon$, 
hence $x\succeq \underline{y}^{(\epsilon)}_\infty$ implies $x \postmaj{D_\infty} y$. Similar arguments can be used to prove the equivalence between $x \premaj{D_\infty}  y$ and 
$\xmax_\infty \succeq y$.
\end{proof}

\hfill

\begin{proof}[{Proof of Theorem} \ref{prop:uppboundspostmaj}:]
	Let us first show the equivalence between both quantities $\epsilon_1$ and $\epsilon_2$.
	On the one hand, from the definition of $\epsilon_1$, one has $x\succeq\underline{y}^{(\epsilon_1)}_\infty$, so that $\epsilon_1 = \|\underline{y}^{(\epsilon_1)}_\infty - y\|_{\infty} \in \{\|y'-y\|_{\infty} \, \colon\, x\succeq y'\}$. Then, $\epsilon_1 \geq \epsilon_2$.
	On the other hand, from the definition of $\epsilon_2$, one has $x \succeq y'_0 \succeq \underline{y}^{(\epsilon_2)}_\infty$ with $y'_0:= \arg\min\{\|y'-y\|_{\infty}\,\colon\,x\succeq y'\}$. Hence, $\epsilon_2\geq \epsilon_1$. Therefore, $\epsilon_1 = \epsilon_2$.
	
	Let us note that, since $x\succeq x$, one has $\epsilon_2 \leq \|x-y\|_{\infty}$. Let us now show that this upper bound is sharp. Take for example $x = \left(\frac{1}{3}, \frac{1}{3}, \frac{1}{3}\right)$ and $y = \left(\frac{1}{2}, \frac{1}{4}, \frac{1}{4}\right)$ (clearly $x \nsucceq y$). Then, $\|x-y\|_{\infty} = \frac{1}{6}$ and it is easy to check $\underline{y}^{(1/6)}_\infty=x$ and $x\nsucceq\underline{y}^{(\epsilon)}_\infty$ for $\epsilon < \frac{1}{6}$, hence $\epsilon_1 = \frac{1}{6}$.
\end{proof}

\hfill

\begin{proof}[{Proof of Theorem} \ref{prop:uppboundspremaj}:]
	Let us first show the equivalence between the quantities $\epsilon_1$ and $\epsilon_2$. On the one hand, from the definition of $\epsilon_1$, one has $\overline{x}^{(\epsilon_1)}_\infty \succeq y$, so that $\epsilon_1=\|\overline{x}^{(\epsilon_1)}_\infty-x\|_{\infty} \in \{\|x'-x\|_{\infty}\,\colon\,x'\succeq y\}$. Then, $\epsilon_1 \geq \epsilon_2$.
	On the other hand, from the definition of $\epsilon_2$, one has $\overline{x}^{(\epsilon_2)}_\infty \succeq x_0'\succeq y$ with $x_0':= \arg\min\{\|x'-x\|_{\infty}\,\colon\,x'\succeq y\}$. Hence, $\epsilon_2\geq \epsilon_1$. Therefore, $\epsilon_1 = \epsilon_2$.
	
	Let us note that, since $y\succeq y$, one has $\epsilon_2 \leq \|x-y\|_{\infty}$. Let us now show that this upper bound is sharp. Take for example $x = \left(\frac{1}{2}, \frac{1}{4}, \frac{1}{4}\right)$ and $y = (1, 0, 0)$ (clearly $x \nsucceq y$). Then, $\|x-y\|_{\infty} = \frac{1}{2}$ and it is easy to check that $\overline{x}^{(1/2)}_\infty = y$ and $\xmax_\infty \nsucceq y$ for $\epsilon < \frac{1}{2}$. Hence, $\epsilon_1 = \frac{1}{2}$.
\end{proof}

\section{Proofs of Sec.~\ref{sec:nonuniform}:}
\label{app:nonuniform}

\begin{proof}[Proof of Theorem \ref{th:equivapproxtransf}:]

Let us first rewrite the statements,
\begin{enumerate}[(i)]
\item there exists a quantum state $\sigma'$ such that $\rho \mapsto \sigma'$ and $\|\sigma -\sigma'\|_p \leq \epsilon$.
\item there exists a probability vector $y'$ such that $x(\rho) \maj y'$ and $\|x(\sigma) - y'\|_p \leq \epsilon$.
\item there exists a quantum state $\rho'$ such that $\rho' \mapsto \sigma$ and $\|\rho -\rho'\|_p \leq \epsilon$.
\item there exists a probability vector $x'$ such that $x' \maj x(\sigma)$ and $\|x(\rho) - x'\|_p \leq \epsilon$.
\end{enumerate}

Let us prove the implication \ref{stii}$\Rightarrow$\ref{sti}.
Let $\sigma = \sum_k x_k(\sigma) \ket{k}\bra{k}$ be the spectral decomposition of $\sigma$
where the basis is sorted in such a way that $x_1(\sigma)\ge\ldots\ge x_d(\sigma)$.
To see that \ref{stii} implies \ref{sti}, one can define the quantum state 
$\sigma'=\sum_k y'_k  \ket{k}\bra{k}$, which is diagonal in the same basis as $\sigma$ and whose eigenvalues are given by the distribution $y'$. Then, $\|\sigma - \sigma'\|_p = \|x(\sigma) - y'\|_p \leq \epsilon$, where the last inequality follows by hypothesis. Finally, due to Lemma~\ref{lemma:nonunifmaj} and the fact that $y'=x(\sigma')$, 
the condition $x(\rho) \succeq y'$ implies that there exists a unital operation such that $\rho \mapsto \sigma'$.

The implication \ref{stiv}$\Rightarrow$\ref{stiii} follows as before by considering diagonal density matrices.
Indeed, let $\rho=\sum_k x_k(\rho)\ket{k}\bra{k}$ be the spectral decomposition of $\rho$ where the basis is sorted in such way that  $x_1(\rho)\ge\ldots\ge x_d(\rho)$.
Now, taking $\rho'=\sum_k x'_k\ket{k}\bra{k}$ and following the same arguments as before one obtains the implication \ref{stiv}$\Rightarrow$\ref{stiii}.

\

In order to prove \ref{sti}$\Rightarrow$\ref{stii} and \ref{stiii}$\Rightarrow$\ref{stiv} we need to recall Lidskii's Theorem (\eg 
\cite[Th. III.4.4]{Bhatia1997}). 
It says that if $\Phi$ is a symmetric gauge function in $\mathbb{R}^d$ (for instance $\Phi(x)=\|x\|_p$ or $\Phi(x)=\|x\|_{\infty}$, see in 
\cite[Ex. II.3.13]{Bhatia1997} for more examples), then
\[
\Phi(x(\rho)-x(\sigma))\le \Phi(x(\rho-\sigma)),
\]
where $\rho$ and $\sigma$ are Hermitian and 
the entries of $x(\rho)$ and $x(\sigma)$ are the eigenvalues of $\rho$ and $\sigma$ sorted in non-increasing order, respectively.
Also, recall the following property of the Schatten $p$-norm,
\[
\|\rho\|_p=\|x(\rho)\|_p.
\]
Then, combining both results, it follows 
\[
\|x(\rho)-x(\sigma)\|_p \le \|\rho-\sigma\|_p.
\]
Hence, \ref{sti}$\Rightarrow$\ref{stii} follows by taking $y'=x(\sigma')$ together with Lemma~\ref{lemma:nonunifmaj}, whereas  
\ref{stiii}$\Rightarrow$\ref{stiv} follows by taking $x'=x(\rho')$ together with Lemma~\ref{lemma:nonunifmaj}.
\end{proof}

\section*{References}

\end{document}